\documentclass[journal,10pt]{IEEEtran}
\usepackage{amsmath}
\usepackage{amssymb}
\usepackage{amsfonts}
\usepackage{graphicx}
\usepackage{epsfig}
\usepackage{subfigure}
\usepackage{psfrag}
\usepackage{color}
\usepackage{url}
\usepackage{cite}
\begin{document}
\title{CoMP Meets Smart Grid: A New Communication and Energy Cooperation Paradigm}

\author{Jie Xu and Rui Zhang
\thanks{This paper has been presented in part at IEEE Global Communications Conference (Globecom), Atlanta, GA USA, December 9-13, 2013.}
\thanks{J. Xu is with the Department of
Electrical and Computer Engineering, National University of
Singapore (email: elexjie@nus.edu.sg).}
\thanks{R. Zhang is with
the Department of Electrical and Computer Engineering, National
University of Singapore (e-mail: elezhang@nus.edu.sg). He is also
with the Institute for Infocomm Research, A*STAR, Singapore.}}
\maketitle

\begin{abstract}\label{sec:abstract}
  In this paper, we pursue a unified study on smart grid and coordinated multi-point (CoMP) enabled wireless communication by investigating a new joint communication and energy cooperation approach. We consider a practical CoMP system with clustered multiple-antenna base stations (BSs) cooperatively communicating with multiple single-antenna mobile terminals (MTs), where each BS is equipped with local renewable energy generators to supply power and {\color{black}also a smart meter to enable two-way energy flow with the grid}. We propose a new \emph{energy cooperation} paradigm, {\color{black}{where a group of BSs dynamically share their renewable energy for more efficient operation via locally injecting/drawing power to/from an aggregator with a zero effective sum-energy exchanged. Under this new energy cooperation model}}, we consider the downlink transmission in one CoMP cluster with cooperative zero-forcing (ZF) based precoding at the BSs. We maximize the weighted sum-rate for all MTs by jointly optimizing the transmit power allocations at cooperative BSs and their exchanged energy amounts subject to a new type of power constraints featuring energy cooperation among BSs with practical loss ratios. Our new setup with BSs' energy cooperation generalizes the conventional CoMP transmit optimization under BSs' sum-power or individual-power constraints. It is shown that with energy cooperation, the optimal throughput is achieved when all BSs transmit with all of their available power, which is different from the conventional CoMP schemes without energy cooperation where BSs' individual power constraints may not be all tight at the same time. This result implies that some harvested energy may be wasted without any use in the conventional setup due to the lack of energy sharing among BSs, whereas the total energy harvested at all BSs is efficiently utilized for throughput maximization with the proposed energy cooperation, thus leading to a new {\emph{energy cooperation gain}}. Finally, we validate our results by simulations under various practical setups, and show that the proposed joint communication and energy cooperation scheme substantially improves the downlink throughput of CoMP systems powered by smart grid and renewable energy, as compared to other suboptimal designs without communication and/or energy cooperation.
\end{abstract}


\IEEEpeerreviewmaketitle
\setlength{\baselineskip}{1\baselineskip}
\newtheorem{definition}{\underline{Definition}}[section]
\newtheorem{fact}{Fact}
\newtheorem{assumption}{Assumption}
\newtheorem{theorem}{\underline{Theorem}}[section]
\newtheorem{lemma}{\underline{Lemma}}[section]
\newtheorem{corollary}{\underline{Corollary}}[section]
\newtheorem{proposition}{\underline{Proposition}}[section]
\newtheorem{example}{\underline{Example}}[section]
\newtheorem{remark}{\underline{Remark}}[section]
\newtheorem{algorithm}{\underline{Algorithm}}[section]
\newcommand{\mv}[1]{\mbox{\boldmath{$ #1 $}}}

\section{Introduction}\label{sec:1}

\PARstart{I}{mproving} energy efficiency in cellular networks has received significant attentions recently. Among assorted energy saving or so-called green techniques that were proposed (see e.g. \cite{Hasan} and the references therein), exploiting renewable energy such as solar and/or wind energy to power cellular base stations (BSs) is a practically appealing solution to reduce the on-grid energy consumption of cellular networks, since the renewable energy is in general more ecologically and economically efficient than conventional energy generated from e.g. fossil fuels \cite{huawei,HanAnsari2014}. Moreover, with recent advancement in smart grid technologies, end users such as BSs in cellular networks {\color{black}can employ smart meters to enable both two-way information and energy flows} with the grid (see e.g. \cite{XueSmartGird,Leithon2013,Chen2013,Ilic,SaadHan} and the references therein) for more efficient and flexible utilization of their locally produced renewable energy that is random and intermittent in nature. {\color{black}However, since there are a large number of BSs in the network, the challenge faced by the cellular operator is how to efficiently coordinate the BSs' renewable generations to match their energy demands, by taking advantage of the two-way information and energy flows in smart grid \cite{XueSmartGird}.}

On the other hand, in order to mitigate the inter-cell interference (ICI) for future cellular networks with more densely deployed BSs, BSs' cooperation or the so-called coordinated multi-point (CoMP) transmission has been extensively investigated in the literature \cite{Gesbert2010,Irmer2011}. With CoMP transmission, BSs share their transmit messages as well as channel state information (CSI) so as to enable cooperative downlink transmissions to mobile terminals (MTs) by utilizing the ICI in a beneficial way for coherent combining. In practice, since the transmit messages and CSI sharing among cooperative BSs are limited by the capacity and latency of backhaul links, a full-scale CoMP transmission by coordinating all the BSs is difficult to implement in practical systems. Therefore, clustered CoMP transmission is more favorable, where BSs are partitioned into different clusters and each cluster implements CoMP transmission separately \cite{ZhangChen2009,NgHuang2010}.

\begin{figure}
\centering
 \epsfxsize=1\linewidth
    \includegraphics[width=8.5cm]{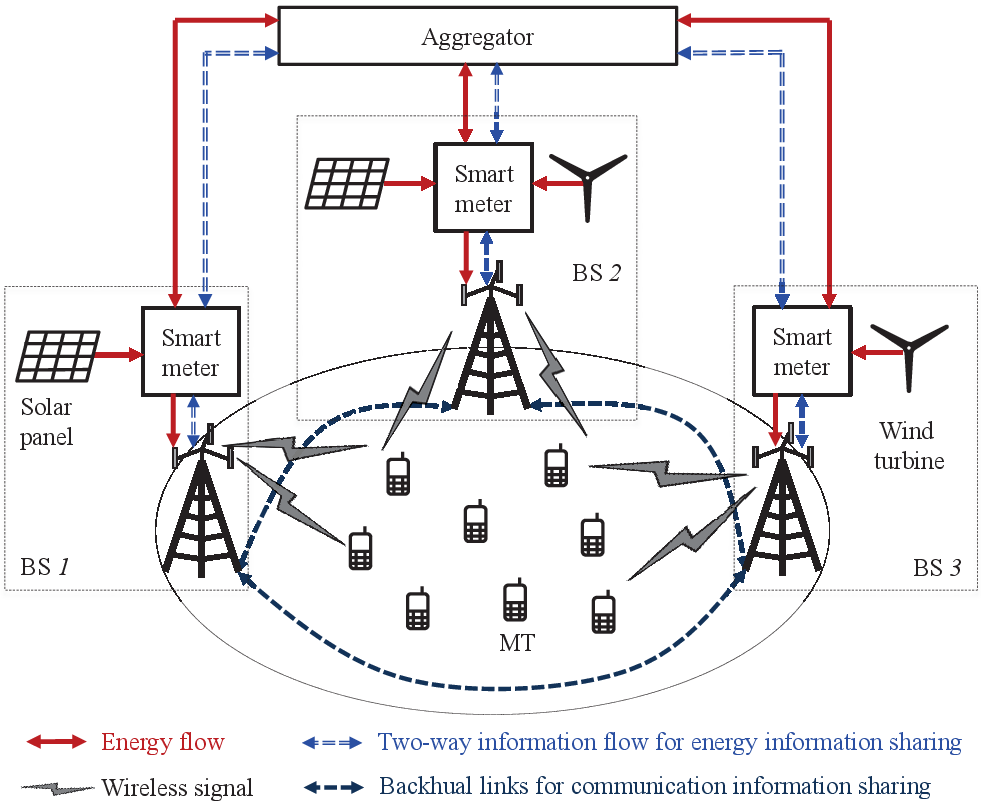}
\caption{An example of a three-cell CoMP system with joint communication and energy cooperation, where the BSs are equipped with local wind and/or solar energy harvesting devices and can share energy among each other through bidirectional energy transfer {\color{black}{through an aggregator}}.} \label{fig:1}
\end{figure}

In this paper, we pursue a unified study on both the smart grid and CoMP enabled cellular networks as shown in Fig. \ref{fig:1}, where each BS is equipped with one or more energy harvesting devices (wind turbines and/or solar panels) to provide renewable energy, and {\color{black}also a smart meter to enable the two-way information and energy flows with the smart grid. To effectively utilize the unevenly generated wind/solar energy over geographically distributed BSs to match their demands, we propose a new {\it energy cooperation} paradigm for the BSs to share their renewable energy with each other. The implementation of energy cooperation is through the aggregator \cite{Gkatzikis2013}, which serves as a mediator or broker between the grid operator and a group of BSs to coordinate the BSs' two-way energy flows. With an aggregator, the energy sharing between any two BSs is realized via one BS locally injecting power to the aggregator and simultaneously the other BS drawing power from it. As a return, the group of BSs need to pay the aggregator a service fee, while the BSs should also commit to ensure that their total power effectively injected into the aggregator is equal to that drawn from the aggregator, in order to maintain the supply-load equilibrium at the aggregator. In practice, the service fee paid to the aggregator should be carefully decided by balancing the trade-off between the profit for the aggregator and the cost saving for the BSs brought by energy cooperation. In this paper, we assume that such service fee is sufficiently low for the cellular operator and thus is ignored for simplicity.}

For the purpose of exposition, we study the joint communication and energy cooperation approach by focusing on one single CoMP cluster, where a group of multiple-antenna BSs cooperatively transmit to multiple single-antenna MTs by applying zero-forcing (ZF) based precoding \cite{ZhangChen2009,Zhang2010}. We jointly optimize the transmit power allocations at cooperative BSs and the amount of transferred energy among them so as to maximize the weighted sum-rate at all MTs, subject to  a new type of power constraints at BSs featuring their energy cooperation with practical loss ratios. Interestingly, our new setup with BSs' energy cooperation can be viewed as a generalization of the conventional CoMP transmit optimization under the BSs' sum-power constraint (e.g., by assuming ideal energy sharing among BSs without any loss) or BSs' individual-power constraints (e.g., without energy sharing among BSs applied) \cite{Zhang2010}. To solve this general problem, we propose an efficient algorithm by applying the techniques from convex optimization.  Based on the optimal solution, it is revealed that with energy cooperation, the maximum weighted sum-rate is achieved when all the BSs transmit with all of their available power, which is different from the conventional CoMP solution without energy cooperation where BSs' individual power constraints may not be all tight at the same time (see Section \ref{sec:comp} for more details). This interesting result implies that some harvested energy may be wasted without any use in the conventional setup due to the lack of energy sharing among BSs, whereas the total energy harvested at all BSs is efficiently utilized to maximize the throughput with the proposed energy cooperation, thus leading to a new {\it energy cooperation gain}. Finally, we validate our results by simulations under various practical setups, and show that the proposed joint communication and energy cooperation scheme substantially improves the downlink throughput of CoMP systems powered by smart grid and renewable energy, as compared to other suboptimal designs without communication and/or energy cooperation.

It is worth noting that exploiting two-way energy flows to help integrate distributed energy prosumers into the smart grid has been actively considered by government regulations (e.g., feed-in tariff and net metering{\footnote{See e.g. {\url{http://en.wikipedia.org/wiki/Feed-in_tariff}}, and {\url{http://en.wikipedia.org/wiki/Net_metering}}.}}), and also attracted significant research interests recently \cite{Leithon2013,Chen2013,Ilic,SaadHan}. {\color{black}For instance, one possible approach is to allow the grid operator to directly coordinate the prosumers by setting time-varying prices for them to buy and sell energy \cite{Leithon2013,Chen2013}. However, this approach may require high complexity and overhead for implementation at the grid operator due to the large number of prosumers such as distributed BSs in the cellular network each with a limited energy supply/demand amount. It may also induce a high energy cost to the cellular operator since the grid operator often sets the energy buying price to be much higher than the selling price to maximize its own revenue. Differently, our proposed energy cooperation through the aggregator is more promising as it ensures win-win benefits for all the parties involved: First, the complexity of implementing the two-way energy flow with the BSs is significantly reduced, since the grid operator only needs to deal with a small number of super-prosumers (BS groups) via the aggregators. Second, each aggregator can gain revenue by charging a service fee to the cellular operator; while energy cooperation of BSs through the aggregator also leads to a lower energy cost of the cellular operator, thanks to the more efficient utilization of the BSs' locally generated renewable energy for saving the expensive on-grid energy purchase (provided that such cost saving well compensates the service fee paid to the aggregator).}

It is also worth noting that there have been other recent works in the literature \cite{Chia2013,GuoTCOM,Zheng2013,GongNiu2013} that investigated another way to implement energy cooperation in cellular networks, where energy exchange is realized by deploying dedicated power lines connecting different BSs. However, this approach may be too costly to be implemented in practice. In contrast, our proposed energy cooperation by utilizing {\color{black}the aggregator and} the existing grid infrastructures is a new solution that is more practically feasible. Moreover, \cite{Gurakan2012,Zheng2014} have proposed to implement the energy exchange among wireless terminals via a technique so-called wireless energy transfer, which, however, has very limited energy transfer efficiency that renders it less useful for BSs' energy sharing in cellular networks. Furthermore, \cite{BuYu2012} has studied smart grid powered cellular networks, in which the utilities of both the cellular network and the power network are optimized based on a two-level Stackelberg game.

The remainder of this paper is organized as follows. Section \ref{sec:system} introduces the system model and presents the problem formulation for joint communication and energy cooperation. Section \ref{sec:optimal} shows the optimal solution to the formulated problem. Section \ref{sec:suboptimal} presents various suboptimal solutions without energy and/or communication cooperation. Section \ref{sec:numerical} provides simulation results to evaluate the performances of proposed optimal and suboptimal schemes. Finally, Section \ref{sec:conclusion} concludes the paper.

{\it Notation:} Scalars are denoted by lower-case letters, vectors by bold-face lower-case letters and matrices by bold-face upper-case letters. $\mv{I}$ and $\mv{0}$ denote an identity matrix and an all-zero matrix, respectively, with appropriate dimensions. $\mathbb{E}(\cdot)$ denotes the statistical expectation. For a square matrix $\mv{S}$, $\mathtt{tr}(\mv{S})$ denotes the trace of $\mv{S}$. For a matrix $\mv{M}$ of arbitrary size, $\mv{M}^H$ and $\mv{M}^T$ denote the conjugate transpose and transpose of $\mv{M}$, respectively. ${\mathtt{Diag}}(x_1, \cdots, x_K)$ denotes a diagonal matrix with the diagonal elements given by $x_1, \cdots, x_K$. $\mathbb{C}^{x\times y}$ denotes the space of $x\times y$ complex matrices.

\section{System Model and Problem Formulation}\label{sec:system}
We consider a practical clustered CoMP system by focusing on one given cluster, in which $N$ BSs each equipped with $M$ antennas cooperatively send independent messages to $K$ single-antenna MTs. {\color{black}As shown in Fig. \ref{fig:1}, the BSs are assumed to be locally deployed with solar panels and/or wind turbines for energy harvesting from the environment, and also equipped with smart meters to enable their energy cooperation through the aggregator in smart grid}. We consider a narrow-band system with ZF-based precoding, which requires that the number of active MTs in each cluster is no larger than the total number of transmitting antennas at all $N$ BSs, i.e., $K\le MN$; while the results of this paper can be readily generalized to more practical setups with broadband transmission and arbitrary number of MTs by applying time-frequency user scheduling (see e.g. \cite{YooGoldsmith}) and/or other precoding schemes. For convenience, we denote the set of MTs and that of BSs as $\mathcal{K} = \{1,\ldots,K\}$ and $\mathcal{N}=\{1,\ldots,N\}$, respectively, with $k,l$ indicating MT index and $i,j$ for BS index.

We assume that all BSs within each cluster can perfectly share their {\color{black}communication information} (including both the transmit messages and CSI) via high-capacity low-latency backhaul links, similarly as in \cite{Zhang2010,ZhangChen2009,NgHuang2010}, {\color{black}and can also perfectly measure and exchange their energy information (i.e., the energy harvesting rates over time) by using the smart meters. Here, the assumption of perfect information sharing enables us to obtain the performance upper bound and characterize the theoretical limits of practical systems.} It is also assumed that BSs in each cluster can {\color{black}implement energy cooperation to} share energy with each other by locally injecting/drawing power to/from the {\color{black}aggregator}. With both information and energy sharing, the joint communication and energy cooperation among $N$ BSs within each cluster can be coordinated by a central unit, which gathers {\color{black}the communication and energy information} from all BSs, and jointly optimizes their cooperative transmission and energy sharing.  The central unit can be either a separate entity deployed in the network, or one of the $N$ BSs which serves as the cluster head.

Furthermore, we assume quasi-static time-slotted models for both renewable energy and wireless channels, where the energy harvesting rates and the channel coefficients remain constant in each slot and may change from one slot to another. In practice, the duration of a communication block is usually on the order of milliseconds due to the practical wireless channel coherence time, while the energy harvesting process evolves at a much slower speed, e.g., solar and wind power typically remains constant over windows of seconds. Without loss of generality, in this paper we choose one communication block as the reference time slot and use the terms ``energy'' and ``power'' interchangeably by normalizing the slot duration. {\color{black}Under this model, the energy harvesting rates at BSs remain constant over e.g. several thousands of communication slots, and thus we focus our study on the joint communication and energy cooperation in one communication slot with given fixed energy harvesting rates at all BSs.} In the following, we first explain the energy cooperation model at BSs, then present the downlink CoMP transmission model, and finally formulate the joint communication and energy cooperation design problem.

\subsection{Energy Cooperation Model}

We consider the energy management at each BS as depicted in Fig. \ref{fig:2}, where the BS does not purchase any expensive on-grid energy to minimize the cost, but instead only uses its locally harvested energy and the transferred energy from other BSs (if not zero) in the same cluster.  {\color{black}The energy harvesting, consumption, and sharing at each BS is coordinated by a smart meter. For energy harvesting,} let the harvested renewable energy at BS $i\in\mathcal N$ be denoted as $E_i\ge 0$. Regarding energy consumption, we only consider the transmit power consumption at each BS by ignoring its non-transmission power due to air conditioner, data processor, etc. for simplicity,{\footnote{Since the non-transmission power is generally modelled as a constant term, our results are readily extendible to the case with non-transmission power included by simply modifying the harvested energy as that offset by the non-transmission related energy consumption at each BS. In this case, the locally generated energy rate should be no smaller than the constant non-transmission power so as to guarantee the reliable operation of BSs. This can be ensured by carefully planning the solar and wind energy harvesting devices at individual BSs or utilizing other backup energy sources such as fuel cells \cite{huawei}.}} and we denote the transmit power of BS $i\in\mathcal N$ as $P_i\ge 0$.

\begin{figure}
\centering
 \epsfxsize=1\linewidth
    \includegraphics[width=8.5cm]{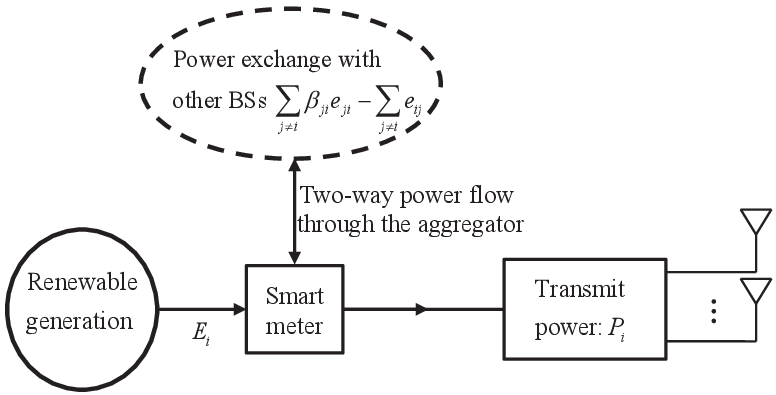}
\caption{Energy management schematics at BS $i$.} \label{fig:2}
\end{figure}

Next, we introduce the energy cooperation among BSs. Let the transferred energy from BS $i$ to BS $j$ be denoted as $e_{ij}\ge 0$, $i,j\in\mathcal{N}, i\neq j$. Practically, this can be implemented by BS $i$ injecting power $e_{ij}$ to the {\color{black}{aggregator}} and at the same time BS $j$ drawing power $e_{ij} - q_{ij}(e_{ij})$ from the {\color{black}{aggregator}}, where $q_{ij}(e_{ij})$ denotes the resulting power loss in the power network with $0<q_{ij}(e_{ij})<e_{ij}$. In practical power systems, the power transfer loss $q_{ij}(e_{ij})$ is normally characterized by a linear model, i.e., $q_{ij}(e_{ij}) = \zeta_{ij} e_{ij}$, where $0 < \zeta_{ij} < 1$ is termed the ``incremental loss''  that represents the incremental power loss in the power network caused by {\color{black}both the power injected by BS $i$ and power drawn by BS $j$}, and $\zeta_{ij}$ is normally calculated via the so-called ``$\mv B$ matrix loss formula'' method (see \cite{WoodWollenberg} for more details). The linear power loss model is practically valid because the variation of power flows in the network due to the injected power $e_{ij}$ is negligibly small as compared to the total power volume in the network. For notational convenience, we define $\beta_{ij} \triangleq 1-\zeta_{ij}$ as the energy transfer efficiency from BS $i$ to BS $j$, where $0 < \beta_{ij} < 1, \forall i,j\in\mathcal{N}, i\neq j$. Thus, when BS $i$ transfers power $e_{ij}$ to BS $j$, the effective energy drawn at BS $j$ needs to be $\beta_{ij}e_{ij}$, in order to have a zero net energy exchanged with the aggregator to maintain its stability. For the ease of analysis, we also consider the special case of $\beta_{ij}=0, \forall i,j\in\mathcal{N}, i\neq j$, which occurs when no energy transfer among BSs is implemented, as well as another special case of $\beta_{ij}=1, \forall i,j\in\mathcal{N}, i\neq j$ for the ideal scenario of no energy transfer loss in the network.

With the aforementioned energy cooperation model, the available transmit power of BS $i$, $P_i$, should satisfy the following constraint:
\begin{align}
0 \le P_i & \le  E_i + \sum_{j\in\mathcal{N},j\neq i}\beta_{ji} e_{ji} - \sum_{j\in\mathcal{N},j\neq i}e_{ij},  i\in\mathcal{N}.\label{eqn:v6:1}
\end{align}

It is worth noting from (\ref{eqn:v6:1}) that to implement the energy exchange with the other $N-1$ BSs, each BS $i\in\mathcal{N}$ only needs to either draw the total amount of energy, $\sum_{j\in\mathcal{N}}\beta_{ji} e_{ji}-\sum_{j\in\mathcal{N}}e_{ij}$, from the {\color{black}{aggregator}} if this term is positive, or otherwise inject the opposite amount (i.e., $-\sum_{j\in\mathcal{N}}\beta_{ji} e_{ji}+\sum_{j\in\mathcal{N}}e_{ij}$) of energy into the {\color{black}{aggregator}} (cf. Fig. \ref{fig:2}). Furthermore, since the total power injected to the grid by all BSs, i.e., $\sum_{i\in\mathcal{N}}\sum_{j\in\mathcal{N},j\neq i}e_{ij}$, is equal to that drawn from the grid by all BSs, i.e., $\sum_{i\in\mathcal{N}}\sum_{j\in\mathcal{N},j\neq i}\beta_{ij}e_{ij}$, plus the total power loss, i.e., $\sum_{i\in\mathcal{N}}\sum_{j\in\mathcal{N},j\neq i}\zeta_{ij}e_{ij}$, the net energy exchanged with the {\color{black}{aggregator}} is zero.

\subsection{Downlink Cooperative ZF Transmission}

We denote the channel vector from BS $i$ to MT $k$ as $\mv{h}_{ik} \in \mathbb{C}^{1\times M}, i\in\mathcal{N},k\in\mathcal{K}$, and the channel vector from all $N$ BSs in one particular cluster to MT $k$ as $\mv{h}_k = [\mv{h}_{1k}~\ldots~\mv{h}_{Nk}]\in \mathbb{C}^{1\times MN}, k\in\mathcal{K}$. It is assumed that the channel vector $\mv{h}_k$'s are drawn from a certain set of independent continuous distributions. Since we consider cooperative downlink transmission by $N$ BSs, the downlink channel in each cluster can be modelled as a $K$-user multiple-input single-output broadcast channel (MISO-BC) with a total number of $MN$ transmitting antennas from all $N$ BSs.

We consider cooperative ZF precoding at BSs \cite{ZhangChen2009,Zhang2010}, with $K\leq MN$,  although the cases of other precoding schemes can also be studied similarly. Let the information signal for MT $k\in\mathcal{K}$ be denoted by $s_k$ and its associated precoding vector across $N$ BSs denoted by $\mv{t}_k \in \mathbb{C}^{MN \times 1}$. Accordingly, the transmitted signal for MT $k$ can be expressed as
$\mv{x}_k = \mv{t}_k s_k.$ Thus, the received signal at MT $k$ is given by
\begin{align}\label{eqn:2}
\mv{y}_k = \mv{h}_k \mv{x}_k + \sum_{l\in\mathcal{K},l\neq k}\mv{h}_k\mv{x}_l + v_k, k\in\mathcal{K},
\end{align}
where $\mv{h}_k \mv{x}_k$ is the desired signal for MT $k$, $\sum_{l\in\mathcal{K},l\neq k}\mv{h}_k\mv{x}_l$ is the inter-user interference within the cluster, and $v_k$ denotes the background additive white Gaussian noise (AWGN) at MT $k$, which is assumed to be of zero mean and variance $\sigma_{k}^2$. Note that in this paper, the background noise $v_k$ at each receiver $k$ may include the downlink interference from other BSs outside the cluster, although this effect can be neglected if proper frequency assignments over different clusters have been designed to minimize any inter-cluster interference. It is also assumed that Gaussian signalling is employed at BSs, i.e., $s_k$'s are circularly symmetric complex Gaussian (CSCG) random variables with zero mean and unit variance. Thus, the covariance matrix of the transmitted signal for MT $k$ can be expressed as $\mv{S}_k = \mathbb{E}(\mv{x}_k\mv{x}_k^H) = \mv{t}_k\mv{t}_k^H \succeq \mv{0}$. Accordingly, the transmit power at BS $i$ can be expressed as \cite{Zhang2010}
\begin{align}\label{eqn:3}
P_i = \sum_{k\in\mathcal{K}}\mathtt{tr}(\mv{B}_i\mv{S}_k),  i\in\mathcal{N},
\end{align}
where \begin{align*}
{\mv{B}}_{i}\triangleq{\mathtt{
Diag}}\bigg(\underbrace{0,\cdots,0}_{(i-1)M},\underbrace{1,\cdots,1}_{M},\underbrace{0,\cdots,0}_{(N-i)M}\bigg).
\end{align*}
By combining (\ref{eqn:3}) and (\ref{eqn:v6:1}), we obtain the transmit power constrains under BSs' energy cooperation, given by
\begin{align}\label{eqn:power_cons}
\sum_{k\in\mathcal{K}}\mathtt{tr}(\mv{B}_i\mv{S}_k) \le  E_i + \sum_{j\in\mathcal{N},j\neq i}\beta_{ji} e_{ji} - \sum_{j\in\mathcal{N},j\neq i}e_{ij},  i\in\mathcal{N}.
\end{align}

Specifically, the cooperative ZF precoding is described as follows. Define $\mv{H}_{-k} = \left[\mv{h}_{1}^{T}, \ldots, \mv{h}_{k-1}^{T}, \mv{h}_{k+1}^{T},\right.$ $\left.\ldots, \mv{h}_{K}^{T}\right]^T$, $k\in\mathcal{K}$, where $\mv{H}_{-k}\in\mathbb{C}^{(K-1)\times MN}$. Let the (reduced) singular value decomposition (SVD) of  $\mv{H}_{-k}$ be denoted as $\mv{H}_{-k} = \mv{U}_{k}\mv{\Sigma}_{k}\mv{V}_{k}^H$, where $ \mv{U}_{k}\in\mathbb{C}^{(K-1)\times (K-1)}$ with $ \mv{U}_{k} \mv{U}_{k}^H =  \mv{U}_{k}^H\mv{U}_{k} = \mv{I}$, $\mv{V}_{k}\in\mathbb{C}^{MN\times (K-1)}$ with $ \mv{V}_{k}^H\mv{V}_{k} = \mv{I}$, and $\mv{\Sigma}_{k}$ is a $(K-1)\times(K-1)$ diagonal matrix. Define the projection matrix $\mv{P}_{k} = \mv{I} - \mv{V}_{k}\mv{V}_{k}^H$. Without loss of generality, we can express $\mv{P}_{k} = \tilde{\mv{V}}_{k}\tilde{\mv{V}}_{k}^{H}$, where $\tilde{\mv{V}}_{k} \in \mathbb{C}^{MN \times (MN-K+1)}$ satisfies $\tilde{\mv{V}}_{k}^{H}\tilde{\mv{V}}_{k}=\mv{I}$ and ${\mv{V}}_{k}^{H}\tilde{\mv{V}}_{k}=\mv{0}$. Note that  $[{\mv{V}}_{k},\tilde{\mv{V}}_{k}]$ forms an $MN\times MN$ unitary matrix.  We then consider the cooperative ZF precoding vector given by
\begin{align}\label{eqn:precoder}
\mv{t}_k = \sqrt{p_{k}} \tilde{\mv{V}}_{k}\frac{(\mv{h}_k\tilde{\mv{V}}_{k})^H}{\|\mv{h}_k\tilde{\mv{V}}_{k}\|},
\end{align}
and accordingly, the transmit covariance matrices can be expressed as
\begin{align}
\mv{S}_{k} = p_{k}\frac{\tilde{\mv{V}}_{k}\tilde{\mv{V}}_{k}^H\mv{h}_k^H\mv{h}_k\tilde{\mv{V}}_{k}\tilde{\mv{V}}_{k}^H}{\|\mv{h}_k\tilde{\mv{V}}_{k}\|^2}, \label{eqn:9}
\end{align}
where $p_{k} \ge 0, \forall k \in \mathcal{K}.$  Note that for simplicity, in this paper we consider separately designed ZF precoding and power allocation as in (\ref{eqn:precoder}) and (\ref{eqn:9}), while our results can be extended to the cases with optimal joint ZF precoding and power allocation (see \cite{Zhang2010}) as well as other precoder designs. Under the above ZF precoding design, the inter-user interference within each cluster can be completely eliminated, i.e., we have $\mv{h}_k\mv{t}_l=0$, or equivalently $\mv{h}_{k}\mv{S}_l\mv{h}_{k}^H = 0, \forall k,l\in\mathcal{K},k\neq l$. As a result, the achievable data rate by the $k$th MT can be expressed as
\begin{align}\label{datarate}
r_k = \log_2\bigg(1+\frac{\mv{h}_{k}\mv{S}_{k}\mv{h}_{k}^{H}}{\sigma^2_k}\bigg) = \log_2(1+a_{k}p_{k}),
\end{align}
where $a_{k} = \frac{\|\mv{h}_{k}\tilde{\mv{V}}_{k}\|^2}{\sigma_{k}^{2}}, \forall k\in\mathcal{K}$. Meanwhile, given $\{\mv{S}_{k}\}$ in (\ref{eqn:9}), the power constraints in (\ref{eqn:power_cons}) can be rewritten as
\begin{align}
\sum_{k\in\mathcal{K}} b_{ik}p_k \le  E_i + \sum_{j\in\mathcal{N},j\neq i}\beta_{ji} e_{ji} - \sum_{j\in\mathcal{N},j\neq i}e_{ij},  i\in\mathcal{N},
\end{align}
where  $b_{ik}=\frac{\mathtt{tr}\left({\mv{B}}_i\tilde{\mv{V}}_{k}\tilde{\mv{V}}_{k}^H\mv{h}_k^H\mv{h}_k\tilde{\mv{V}}_{k}\tilde{\mv{V}}_{k}^H\right)}{\|\mv{h}_k\tilde{\mv{V}}_{k}\|^2}$, $\forall i\in\mathcal{N},k \in \mathcal{K}$.{\footnote{Since it is assumed that the channel vector $\mv{h}_k$'s are independently distributed, without loss of generality we assume $a_k>0,b_{ik}>0,\forall i,k$.}}

\subsection{Problem Formulation}

We aim to jointly optimize the transmit power allocations $\{p_{k}\}$ at all $N$ BSs, as well as their transferred energy $\{e_{ij}\}$, so as to maximize the weighted sum-rate throughput (in bps/Hz) for all $K$ MTs given by $\sum\limits_{k\in\mathcal{K}}  \omega_k r_k$, where $r_k$ denotes the achievable rate by the $k$th MT given in (\ref{datarate}), and $\omega_k > 0$ denotes the given weight for MT $k, k\in\mathcal{K}$. Here, larger weight value of $\omega_k$ indicates higher priority of transmitting information to MT $k$ as compared to other MTs; hence, by properly designing the weight $\omega_k$'s, rate fairness among different MTs can be ensured. Next, we formulate the joint communication and energy cooperation problem as
\begin{align}
\mathrm{(P1)}:\nonumber\\
\mathop{\mathtt{max}}\limits_{\{p_k\},\{e_{ij}\}}&  \sum\limits_{k\in\mathcal{K}} \omega_k\log_2(1+a_{k}p_{k})\label{eqn:p1}\\
\mathtt{s.t.}~~&\sum_{k\in\mathcal{K}} b_{ik}p_k \le  E_i + \sum_{j\in\mathcal{N},j\neq i}\beta_{ji} e_{ji} - \sum_{j\in\mathcal{N},j\neq i}e_{ij},\forall i\in\mathcal{N}\label{eqn:8}\\
&p_{k}\ge 0,\ e_{ij}\ge 0,\ \forall k\in\mathcal{K},~i,j\in\mathcal{N}, i\neq j.\label{eqn:11}
\end{align}

Before we proceed to solving problem (P1), we first consider the following two special cases for (P1) with $\beta_{ij} = 0, \forall i,j\in \mathcal{N}, i\neq j$ and $\beta_{ij}=1, \forall i,j\in\mathcal{N}, i\neq j$, respectively, to draw some insights. The former case of $\beta_{ij} = 0,\forall i,j\in \mathcal{N}, i\neq j$ corresponds to the scenario of no energy transfer between any two BSs in the cluster, for which it is equivalent to setting $e_{ij} = 0, \forall i,j\in\mathcal{N}, i\neq j$. In this case, the power constraints in 
(\ref{eqn:8}) reduce to
\begin{align}
\sum_{k\in\mathcal{K}} b_{ik}p_k& \le  E_i,\forall i\in\mathcal{N}.\label{eqn:ver13:power:perBS2}
\end{align}
Accordingly, (P1) reduces to the conventional MISO-BC weighted sum-rate maximization problem with $N$ per-BS power constraints, given by $E_i, \forall i\in\mathcal{N}$ \cite{Zhang2010}. On the other hand, if $\beta_{ij}= 1,\forall i,j\in\mathcal{N}, i\neq j$, the energy of any two BSs can be shared ideally without any loss. It can then be easily verified that in this case BSs' individual power constraints in 
(\ref{eqn:8}) can be combined into one single sum-power constraint as \begin{align}
\sum_{k\in\mathcal{K}}p_k &\le \sum_{i\in\mathcal{N}}E_i.\label{eqn:ver13:sumpower2}
\end{align}
Thus, (P1) reduces to the conventional MISO-BC weighted sum-rate maximization problem under one single sum-power constraint \cite{Zhang2010}.
\begin{figure}
\centering
 \epsfxsize=1\linewidth
    \includegraphics[width=8.5cm]{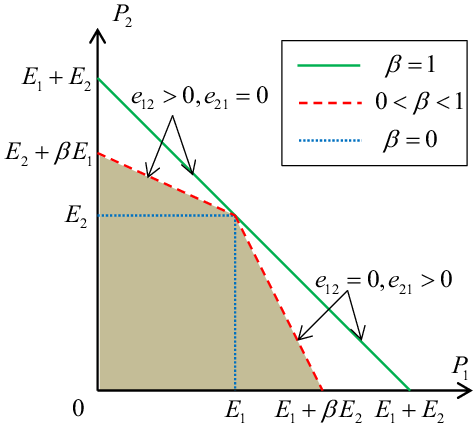}
\caption{Feasible power region of BSs' transmit power under energy cooperation for a two-cell system.} \label{fig:3}
\end{figure}

As an example, we consider a two-cell system with $N=2$, where $\beta_{ij} \triangleq \beta, \forall i,j\in\{1,2\},i\neq j$. We depict in Fig. \ref{fig:3} the feasible power region (shown as  the shaded area) consisting of the all available transmit power pairs $(P_1,P_2)$ at BS 1 and BS 2 under energy cooperation in the case of $0<\beta<1$, as compared to the other two extreme cases of $\beta=0$ and $\beta=1$. It is observed that the boundary of the power region when $\beta>0$ is always attained either by $e_{12}\ge0$, $e_{21}=0$ or  $e_{12}=0$, $e_{21}\ge0$. Similarly, it can be shown that for the general multi-cell case with $N> 2$, the boundary of the $N$-dimension power region should also be attained either by $e_{ij}\ge0$, $e_{ji}=0$ or  $e_{ij}=0$, $e_{ji}\ge0, \forall i,j\in\mathcal{N}, i\neq j$ (for any pair of two BSs $i$ and $j$). This result together with the fact that the weighted sum-rate in (\ref{eqn:p1}) always increases with the transmit power implies that only unidirectional energy transfer between any pair of two BSs is needed to achieve the optimal solution of (P1), as will be more rigorously proved in Section \ref{sec:energyexchange} (see Proposition \ref{proposition:3.3}). It is also observed from Fig. \ref{fig:3} that the power region in the case of $0<\beta<1$ is larger than that without energy sharing, i.e., $\beta=0$, and at the same time smaller than that with ideal energy sharing, i.e., $\beta=1$. Similarly, it is conjectured that for the general multi-cell scenario with $N>2$, the feasible power region in the case of $0<\beta_{ij}<1, \forall i,j\in\mathcal{N},i\neq j$ is also larger than that of $\beta_{ij}=0, \forall i,j\in\mathcal{N},i\neq j$, and at the same time smaller than that of $\beta_{ij}=1, \forall i,j\in\mathcal{N},i\neq j$. As a result, the optimal value of (P1) under practical energy cooperation with $0<\beta_{ij}<1, \forall i,j\in\mathcal{N},i\neq j$ should lie between those of the two extreme cases with the per-BS power constraints ($\beta_{ij}=0, \forall i,j\in\mathcal{N},i\neq j$) and the BSs' sum-power constraint ($\beta_{ij}=1, \forall i,j\in\mathcal{N},i\neq j$), respectively, as will be shown rigorously later in this paper.

\section{Optimal Solution}\label{sec:optimal}

In this section, we solve problem (P1) for the general case of $0\le\beta_{ij}\le1,\forall i,j\in\mathcal{N},i\neq j$ to obtain the optimal joint power allocation and energy transfer solution for BSs' joint communication and energy cooperation.

It can be verified that (P1) is a convex problem, since the objective function is concave and all the constraints are affine. Thus, the Lagrange duality method can be applied to solve this problem optimally \cite{BoydVandenberghe2004}. Let ${\mu}_i \ge 0, i\in\mathcal{N}$, be the dual variable associated with each of the $N$ power constraints of problem (P1) given in (\ref{eqn:8}). Then the partial Lagrangian of problem (P1) can be expressed as
\begin{align}
&\mathcal{L}\left(\{{\mu}_i\},\{{p}_k\},\{e_{ij}\}\right) \nonumber\\=& \sum\limits_{k\in\mathcal{K}}
\omega_k\log_2(1+a_{k}p_{k}) \nonumber\\ &- \sum\limits_{i\in\mathcal{N}}\mu_{i}\left( \sum_{k\in\mathcal{K}} b_{ik} {p}_{k}  -  E_i - \sum_{j\in\mathcal{N},j\neq i}\beta_{ji}e_{ji} + \sum_{j\in\mathcal{N},j\neq i}e_{ij}\right) \nonumber\\
=& \sum\limits_{k\in\mathcal{K}}\left(\omega_k
\log_2(1+a_{k}p_{k}) - \sum\limits_{i\in\mathcal{N}}b_{ik}\mu_{i} {p}_{k} \right)  \nonumber\\ &+\sum\limits_{i,j\in\mathcal{N},i\neq j}(\beta_{ij}\mu_j-\mu_i)e_{ij} + \sum\limits_{i\in\mathcal{N}}\mu_{i}  E_{i}.
\label{eqn:14}
\end{align}
Accordingly, the dual function is given by
\begin{align}
f(\{{\mu}_i\})= \mathop\mathtt{max}_{\{{p}_k\},\{e_{ij}\}} ~& \mathcal{L}\left(\{{\mu}_i\},\{{p}_k\},\{e_{ij}\}\right)\nonumber\\
\mathtt{s.t.}~& {p}_{k}\ge 0, e_{ij}\ge 0,\ \forall k\in\mathcal{K}, i,j\in\mathcal{N}, i\neq j.
 \label{eqn:15}
\end{align}
Thus, the dual problem is defined as
\begin{align}
\mathrm{(P1D)}:\mathop\mathtt{min}_{\{{\mu}_i\ge 0\}}  ~&   f(\{{\mu}_i\}).\label{eqn:Fc:reformulation:dual}
\end{align}
Since (P1) is convex and satisfies the Salter's condition, strong duality holds between (P1) and its dual problem (P1D) \cite{BoydVandenberghe2004}. Thus, we can solve (P1) by solving its dual problem (P1D) equivalently. To solve (P1D), in the following we first solve problem (\ref{eqn:15}) to obtain $f(\{{\mu}_i\})$ with a given set of $\mu_i\geq 0, i\in\mathcal{N}$, and then search over $\{\mu_i\}$ in $\mathbb{R}_N^+$ to minimize $f(\{{\mu}_i\})$ in (\ref{eqn:Fc:reformulation:dual}).

We first give the following lemma.

\begin{lemma}\label{proposition:2}
In order for $f(\{{\mu}_i\})$ to be bounded from above, we have:
\begin{enumerate}
  \item At least one of $\mu_{i}, \forall i\in\mathcal{N}$, must be strictly positive;
  \item $\beta_{ij}\mu_j-\mu_i \le 0, \forall i,j\in\mathcal{N}, i\neq j$, must be true.
\end{enumerate}
\end{lemma}
\begin{proof}
See Appendix \ref{appendix:for:lemma1}.
\end{proof}

According to Lemma \ref{proposition:2}, we only need to solve problem (\ref{eqn:15}) with given $\{{\mu}_i\}$ satisfying $\mu_i \geq 0, \forall i\in\mathcal{N}$ (but not all equal to zero), and $\beta\mu_j-\mu_i \le 0, \forall i,j\in\mathcal{N}, i\neq j$, since otherwise $f(\{{\mu}_i\})$ will be unbounded from above and thus need not to be considered for the minimization problem in (\ref{eqn:Fc:reformulation:dual}). In this case, problem (\ref{eqn:15}) can be decomposed into the following $K+N^2-N$ number of subproblems (by removing the irrelevant constant term $\sum\limits_{i\in\mathcal{N}}\mu_{i}  E_{i}$ in (\ref{eqn:14})):
\begin{align}
\mathop\mathtt{max}\limits_{p_{k}\ge0} ~&\omega_k\log_2(1+a_{k}p_{k}) - \sum\limits_{i\in\mathcal{N}}b_{ik}\mu_{i} {p}_{k}, ~\forall k\in\mathcal{K},\label{eqn:17}\\
\mathop\mathtt{max}\limits_{e_{ij}\ge 0}~ &(\beta_{ij}\mu_j-\mu_i)e_{ij}, ~\forall i,j\in\mathcal{N}, i\neq j.\label{eqn:18}
\end{align}

For the $K$ subproblems in (\ref{eqn:17}), note that $\sum\limits_{i\in\mathcal{N}}b_{ik}\mu_{i}   > 0$ always holds due to Lemma \ref{proposition:2}. Thus, it can be easily verified that the optimal solutions can be obtained as
\begin{align}
p_k^{(\{\mu_i\})} = \left(\frac{\omega_k}{\ln2\sum\limits_{i\in\mathcal{N}} b_{ik}\mu_{i}} - \frac{1}{a_k}\right)^+,  ~\forall k\in\mathcal{K},\label{eqn:pj}
\end{align}
where $(x)^+ = \mathtt{max}(x,0)$. Next, consider the remaining $N^2-N$ subproblems in (\ref{eqn:18}), for which it can be easily shown that  with $\beta_{ij}\mu_j-\mu_i \le 0, \forall i\neq j$ given in Lemma \ref{proposition:2}, the optimal solutions are
\begin{align}\label{eqn:x_i}
e_{ij}^{(\{\mu_i\})} = 0,  ~\forall i,j\in\mathcal{N},i\neq j.
\end{align}
Notice that if $\beta_{ij}\mu_j-\mu_i = 0$, then the optimal solution of $e_{ij}$ to problem (\ref{eqn:18}) is not unique and can take any non-negative value. In this case, we let $e_{ij}^{(\{\mu_i\})}= 0$ given in (\ref{eqn:x_i}) only for solving the dual problem in (\ref{eqn:15}), which may not be the optimal primary solution of $e_{ij}$ to problem (P1).

Combining the results in (\ref{eqn:pj}) and (\ref{eqn:x_i}), we obtain $f(\{{\mu}_i\})$ with given $\{{\mu}_i\}$ satisfying $\mu_i \geq 0, \forall i\in\mathcal{N}$ (but not all equal to zero), and $\beta_{ij}\mu_j-\mu_i \le 0, \forall i,j\in\mathcal{N}$. Then, we solve problem (P1D) in (\ref{eqn:Fc:reformulation:dual}) by finding the optimal $\{{\mu}_i\}$ to minimize $f(\{{\mu}_i\})$. According to Lemma \ref{proposition:2}, (P1D) can be equivalently re-expressed as
\begin{align}
\mathrm{(P1D)}:\mathop\mathtt{min}_{\{\mu_i\}}  ~&   f(\{{\mu}_i\})\nonumber\\
\mathtt{s.t.}  ~&{\mu}_i\ge 0,  ~\forall i\in\mathcal{N}\nonumber\\
&\beta_{ij}\mu_j-\mu_i \le 0,~~ \forall i,j\in\mathcal{N},i\neq  j.\label{eqn:Fc:reformulation:dual2}
\end{align}

Since (P1D) is convex but not necessarily differentiable, a subgradient based method such as the ellipsoid method \cite{Boyd:ConvexII} can be applied, for which it can be shown that  the subgradients of $f(\{{\mu}_i\})$ for given $\mu_i$ are $E_i - \sum_{k\in\mathcal{K}} b_{ik} {p}_{k}^{(\{\mu_i\})} + \sum_{j\in\mathcal{N},j\neq i}\beta_{ji} e_{ji}^{(\{\mu_i\})} - \sum_{j\in\mathcal{N},j\neq i}e_{ij}^{(\{\mu_i\})} = E_i - \sum_{k\in\mathcal{K}} b_{ik} {p}_{k}^{(\{\mu_i\})}, \forall i\in\mathcal{N}$, where the equality follows from (\ref{eqn:x_i}). Therefore, the optimal solution of (P1D) can be obtained as $\{{\mu}_i^*\}$.

With the optimal dual solution $\{{\mu}_i^*\}$ at hand, the corresponding $\{p_k^{(\{\mu_i^*\})}\}$ in (\ref{eqn:pj}) become the optimal solution for (P1), denoted by $\{p_k^*\}$. Now, it remains to obtain the optimal solution of $\{e_{ij}\}$ for (P1), denoted by $\{e_{ij}^*\}$. In general, $\{e_{ij}^*\}$ cannot be directly obtained from (\ref{eqn:x_i}) with given $\{\mu_i^*\}$, since the solution of (\ref{eqn:x_i}) is not unique if $\beta_{ij} \mu_j^*-\mu_i^* = 0$. Fortunately, it can be shown that given $\{p_k^*\}$, any  $\{e_{ij}\}$ that satisfy the linear constraints in (\ref{eqn:8}) and (\ref{eqn:11}) are the optimal solution to (P1). Thus, we can obtain $\{e_{ij}^*\}$ by solving the following feasibility problem.
\begin{align}
\mathop{\mathtt{find}} ~&  \{e_{ij}\}\nonumber\\
\mathtt{s.t.}~~&\sum_{k\in\mathcal{K}} b_{ik}p_k^* \le  E_i + \sum_{j\in\mathcal{N},j\neq i}\beta_{ji} e_{ji} - \sum_{j\in\mathcal{N},j\neq i}e_{ij},\forall i\in\mathcal{N}\nonumber\\
& e_{ij}\ge 0,\ \forall i,j\in\mathcal{N}, i\neq j.\label{eqn:e_ij}
\end{align}
Since problem (\ref{eqn:e_ij}) is a simple linear program (LP), it can be efficiently solved by existing softwares, e.g., $\mathtt{CVX}$ \cite{cvx}. As a result, we have finally obtained $\{e_{ij}^*\}$ and thus have solved (P1) completely.

In summary, one algorithm to solve (P1) for the general case of $0\le\beta_{ij}\le 1, \forall i,j\in\mathcal{N},i\neq j,$ is given in Table I as Algorithm 1, in which $\mu_i = \mu > 0, \forall i\in\mathcal{N}$, are chosen as the initial point for the ellipsoid method in order to satisfy the constraints in (\ref{eqn:Fc:reformulation:dual2}). Note that Algorithm 1 needs to be implemented at a central unit, which is assumed to have all the CSI and energy information from all BSs in the same cluster.

\begin{table}[t]
\begin{center}
\caption{{\bf Algorithm 1}: Algorithm for Solving Problem (P1)} \vspace{0.01cm}
\hrule \vspace{0.01cm}
\begin{itemize}
\item[a)] Initialize $\mu_i = \mu > 0, \forall i\in\mathcal{N}$.
\item[b)] {\bf Repeat:}
    \begin{itemize}
    \item[1)] Obtain $\{p_k^{(\{\mu_i\})}\}$ using (\ref{eqn:pj}) with given  $\{\mu_i\}$;
    \item[2)] Compute the subgradients of $f(\{{\mu}_i\})$ as $E_i - \sum_{k\in\mathcal{K}} b_{ik} {p}_{k}^{(\{\mu_i\})}, \forall i\in\mathcal{N}$, then update $\{\mu_i\}$ accordingly based on the ellipsoid method \cite{Boyd:ConvexII}, subject to the constraints of ${\mu}_i\ge 0,\forall i\in\mathcal{N}$, and $\beta\mu_j-\mu_i \le 0, \forall i,j\in\mathcal{N},i\neq  j.$
    \end{itemize}
\item[c)] {\bf Until} $\{\mu_i\}$ all converge within a prescribed accuracy.
\item[d)] Set $p_k^* = p_k^{(\{\mu_i\})}, \forall k\in\mathcal{K}$.
\item[e)] Compute $\{e_{ij}^*\}$ by solving the LP in (\ref{eqn:e_ij}).
\end{itemize}
\vspace{0.01cm} \hrule \vspace{0.01cm}\label{algorithm:1}
\end{center}
\vspace{-0cm}
\end{table}

\subsection{Practical Implementation of Energy Cooperation}

So far, we have solved (P1) with any given $0\le \beta_{ij}\le 1,\forall i,j\in\mathcal{N},i\neq j$. However, as we have discussed in Section \ref{sec:system}, in order to practically implement the energy cooperation, it may not be necessary for each BS $i$ to know the exact values of $e_{ij}^*$'s or $e_{ji}^*$'s, $\forall j\neq i$, while it suffices to know the total power that should be drawn or injected from/to the {\color{black}aggregator} at BS $i$, i.e., the value of $\sum_{j\in\mathcal{N},j\neq i}\beta_{ji} e_{ji}^* - \sum_{j\in\mathcal{N},j\neq i}e_{ij}^*$. In the following, we focus our study on the practical case when the energy cooperation is feasible within this cluster, i.e., $0<\beta_{ij}< 1,\forall i,j\in\mathcal{N},i\neq j$, under which we derive $\sum_{j\in\mathcal{N},j\neq i}\beta_{ji} e_{ji}^* - \sum_{j\in\mathcal{N},j\neq i}e_{ij}^*, \forall i\in\mathcal{N}$ and accordingly show how to implement the energy cooperation with lower complexity without the need of solving the LP in (\ref{eqn:e_ij}) to obtain $e_{ij}^*$'s.

\begin{proposition}\label{proposition:3.2}
If energy cooperation is feasible within each cluster, i.e., $0<\beta_{ij}< 1,\forall i,j\in\mathcal{N},i\neq j$, then it follows that $\mu_i^* > 0,\forall i\in\mathcal{N}$.
\end{proposition}
\begin{proof}
See Appendix \ref{appendix:1}.
\end{proof}

From Proposition \ref{proposition:3.2} and the following complementary slackness conditions \cite{BoydVandenberghe2004} satisfied by the optimal solution of (P1):
\begin{align}
\mu_{i}^*\left( \sum_{k\in\mathcal{K}} b_{ik} {p}_{k}^*  -  E_i - \sum_{j\in\mathcal{N},j\neq i}\beta_{ji} e_{ji}^* + \sum_{j\in\mathcal{N},j\neq i}e_{ij}^*\right) = 0, \forall i\in\mathcal{N},
\label{eqn:v4:13}\end{align}
it follows that in the case of $0<\beta_{ij}< 1,\forall i,j\in\mathcal{N},i\neq j$, we have
\begin{align}
\sum_{k\in\mathcal{K}} b_{ik} {p}_{k}^*  -  E_i - \sum_{j\in\mathcal{N},j\neq i}\beta_{ji} e_{ji}^* + \sum_{j\in\mathcal{N},j\neq i}e_{ij}^* = 0, \forall i\in\mathcal{N}.
\label{eqn:v4:14}\end{align}
In other words, the optimal solutions of (P1) are always attained with all the power constraints in (\ref{eqn:8}) being met with equality, i.e., the total energy available at all BSs is efficiently utilized to maximize the weighted sum-rate.

From (\ref{eqn:v4:14}), it is evident that once the optimal transmit power allocation $p_k^*$'s are determined, the total amount of power that should be drawn or injected from/to the grid at each BS $i$ can be easily obtained as $\sum_{j\in\mathcal{N},j\neq i}\beta_{ji} e_{ji}^* - \sum_{j\in\mathcal{N},j\neq i}e_{ij}^* = \sum_{k\in\mathcal{K}} b_{ik} {p}_{k}^*  -  E_i, \forall i\in\mathcal{N}$. That is, if the required power for transmission is larger than the
available energy at BS $i$, i.e., $\sum_{k\in\mathcal{K}} b_{ik} {p}_{k}^*  >  E_i$, then an amount of energy, $\sum_{k\in\mathcal{K}} b_{ik} {p}_{k}^*  - E_i > 0$, will be drawn from the {\color{black}aggregator}; otherwise, the extra amount of energy, $E_i - \sum_{k\in\mathcal{K}} b_{ik} {p}_{k}^* > 0$, will be injected to the {\color{black}aggregator}. As a result, it suffices for each BS to know the transmit power allocation $p_k^*$'s to implement the proposed energy cooperation, and thus it is not necessary for the central unit to first solve the LP in (\ref{eqn:e_ij}) and then inform the BSs of the values of $e_{ij}^*$'s. Therefore, the computation complexity at the central unit as well as the signalling overhead in the backhaul can be both reduced.

\subsection{Energy Exchange Behavior of BSs}\label{sec:energyexchange}

Although the exact values of $e_{ij}^*$'s are not required for implementing the energy cooperation, we further discuss them under the practical case of $0<\beta_{ij}< 1, \forall i,j\in\mathcal{N}, i\neq j$ to give more insight on the optimal energy exchange among BSs.

\begin{proposition}\label{proposition:3.3}
If $0<\beta_{ij}< 1, \forall i,j\in\mathcal{N}, i\neq j$, then it must hold that at least one of $\sum_{j\in\mathcal{N},j\neq i}\beta_{ji} e_{ji}^*$ and $\sum_{j\in\mathcal{N},j\neq i}e_{ij}^*$ should be zero, $\forall i\in\mathcal{N}$.
\end{proposition}
\begin{proof}
See Appendix \ref{appendix:2}.
\end{proof}

From Proposition \ref{proposition:3.3} together with the fact that $e_{ij}^* \ge 0,\forall i,j\in\mathcal{N},i\neq j$, it is inferred that in the case of $0<\beta_{ij}< 1, \forall i,j\in\mathcal{N}, i\neq j$, each BS should either transfer power to the other  $N-1$ BSs in the same cluster or receive power from them, but not both at the same time, which affirms our conjecture in Section \ref{sec:system} (cf. discussion for Fig. \ref{fig:3}). This result is quite intuitive, since receiving and transferring energy simultaneously at one BS will inevitably induce unnecessary energy loss.

\subsection{Comparison to No Energy Cooperation}\label{sec:comp}

Finally, we provide a comparison between the case with energy cooperation with $0<\beta_{ij}< 1,\forall i,j\in\mathcal{N},i\neq j,$ versus that without energy cooperation with $\beta_{ij}=0,\forall i,j\in\mathcal{N},i\neq j$, to show the benefit of energy cooperation in renewable powered CoMP systems.

Consider the case of $\beta_{ij}=0,\forall i,j\in\mathcal{N},i\neq j$, i.e., without energy cooperation. In this case, since it is known that $e_{ij} = 0,\forall i,j\in\mathcal{N},i\neq j$ is optimal for (P1), (P1) is simplified as
\begin{align}
\mathrm{(P1-NoEC)}:\mathop{\mathtt{max}}\limits_{\{p_k\}}~&  \sum\limits_{k\in\mathcal{K}}
\log_2(1+a_{k}p_{k})\nonumber\\
\mathtt{s.t.}~
&\sum_{k\in\mathcal{K}}b_{ik}p_{k}\le  E_i,\forall i\in\mathcal{N}\label{eqn:no:energy:coop2}\\
&{p}_{k}\ge 0,\ \forall k\in\mathcal{K}.\nonumber
\end{align}
In problem (P1-NoEC), it can be shown that the transmit power constraints at $N$ BSs in (\ref{eqn:no:energy:coop2}) may not be all tight in the optimal solution, especially when the available energy amounts at different BSs are substantially different. {\color{black}{As an extreme case, consider the setup with $N=2$ single-antenna ($M=1$) BSs cooperatively transmitting to $K=2$ MTs,}} in which the energy harvesting rate at BS 1 (e.g., with wind turbines) is much larger than that at BS 2 (e.g., with solar panels at night), i.e., $E_1 \gg E_2$. In this case, it is likely that for any feasible power allocation $\{p_{k}\}$ satisfying the power constraint of BS 2, i.e., $\sum_{k\in\mathcal{K}}b_{2k}p_{k}\le  E_2$, it holds that $\sum_{k\in\mathcal{K}}b_{1k}p_{k} <  E_1$ due to $E_1 \gg E_2$. Accordingly, the power constraint at BS 1 cannot be tight at the optimal solution, {\color{black}{as will be shown later in the simulation results (cf. Fig. \ref{fig:4})}}. From this example, it can be inferred that with only communication cooperation among BSs but without energy cooperation, the harvested energy at some BSs (those with large harvested energy amount) may not be totally utilized and thus wasted\footnote{Another possible solution to reduce energy waste at each BS is to deploy energy storage device to store unused energy for future use, which however increases the operation cost of the cellular systems.}; as a result, the achievable weighted sum-rate of MTs will be greatly limited by the BSs with less available energy. In contrast, if there are both communication and energy cooperation in the case of   $0<\beta_{ij}< 1,\forall i,j\in\mathcal{N},i\neq j$, it then follows from (\ref{eqn:v4:14}) that  the total energy at all BSs will be used to maximize the weighted sum-rate (albeit that some energy is lost in the energy transfer), which thus results in a new {\it energy cooperation gain} over the conventional CoMP system, as will also be shown by our numerical results in Section \ref{sec:numerical}.

\section{Suboptimal Solutions}\label{sec:suboptimal}

In this section, we consider three suboptimal solutions which apply only energy or communication cooperation or neither of them in a single-cluster CoMP system  with renewable power supplies as assumed in our system model in Section \ref{sec:system}. These solutions will provide performance benchmarks for our proposed optimal solution based on joint communication and energy cooperation.
\subsection{No Energy Cooperation, Communication Cooperation Only}

This scheme corresponds to setting $\beta_{ij}=0,\forall i,j\in\mathcal{N},i\neq j,$ in (P1), where no energy transfer among BSs is feasible. Algorithm 1 can be directly applied to solve the weighted sum-rate maximization problem in this case.

\subsection{No Communication Cooperation, Energy Cooperation Only}\label{subsec:noCom}
For the case with energy cooperation only, we consider the practical case of $0<\beta_{ij}\leq 1,\forall i,j\in\mathcal{N},i\neq j$. For this case, we assume that each BS $i$ serves $K_i$ associated MTs,{\footnote{The MT-BS association is assumed to be given here, which can be obtained by e.g. finding the BS that has the strongest channel to each MT.}} where $\sum_{i\in\mathcal{N}} K_i =K$ and $K_i \le M, \forall i\in\mathcal{N}$. Let $\mathcal{K}_i = \{\sum_{j=1}^{i-1}K_j+1, \ldots, \sum_{j=1}^{i}K_j\}$ denote the set of MTs assigned to BS $i$, where $|\mathcal{K}_i| = K_i$ and $\bigcup_{i\in\mathcal{N}}\mathcal{K}_i=\mathcal{K}$. Each BS then transmits to its associated $K_i$ MTs independently with ZF precoding over orthogonal bands each with ${1}/{N}$ portion of the total bandwidth (to avoid any ICI within each cluster). Although no communication cooperation is applied, BSs can still implement energy cooperation as described in Section \ref{sec:system}.

Let the transmit covariance from BS $i$ to its associated MT $k$ be denoted by $\bar{\mv{S}}_{k} \in \mathbb{C}^{M\times M}, \forall k\in\mathcal{K}_i,i\in\mathcal{N}$. We then design the ZF precoding for each BS $i$ as follows. Define $\bar{\mv{H}}_{-k} = \left[\mv{h}_{i,\sum_{j=1}^{i-1}K_j+1}^{T}, \ldots, \mv{h}_{i, k-1}^{T}, \mv{h}_{i,k+1}^{T},\ldots, \mv{h}_{i,\sum_{j=1}^{i}K_j}^{T}\right]^T$, where $\bar{\mv{H}}_{-k}\in\mathbb{C}^{(K_i-1)\times M}, \forall k\in\mathcal{K}_i, i\in\mathcal{N}$. Let the (reduced) SVD of  $\bar{\mv{H}}_{-k}$ be denoted as $\bar{\mv{H}}_{-k} = \bar{{\mv{U}}}_{k}\bar{\mv{\Sigma}}_{k}\bar{\mv{V}}_{k}^H$, where $\bar{\mv{U}}_{k}\in\mathbb{C}^{(K_i-1)\times (K_i-1)}$ with $\bar{\mv{U}}_{k} \bar{\mv{U}}_{k}^H =  \bar{\mv{U}}_{k}^H\bar{\mv{U}}_{k} = \mv{I}$, $\bar{\mv{V}}_{k}\in\mathbb{C}^{M\times (K_i-1)}$ with $\bar{\mv{V}}_{k}^H\bar{\mv{V}}_{k} = \mv{I}$, and $\bar{\mv{\Sigma}}_{k}$ is a $(K_i-1)\times(K_i-1)$ diagonal matrix. Define the projection matrix $\bar{\mv{P}}_{k} = \mv{I} - \bar{\mv{V}}_{k}\bar{\mv{V}}_{k}^H$. Without loss of generality, we can express $\bar{\mv{P}}_{k} = \hat{\mv{V}}_{k}\hat{\mv{V}}_{k}^{H}$, where $\hat{\mv{V}}_{k} \in \mathbb{C}^{M \times (M-K_i+1)}$ satisfies $\hat{\mv{V}}_{k}^{H}\hat{\mv{V}}_{k}=\mv{I}$ and $\bar{{\mv{V}}}_{k}^{H}\hat{\mv{V}}_{k}=\mv{0}$. Note that  $[\bar{\mv{V}}_{k},\hat{\mv{V}}_{k}]$ forms an $M\times M$ unitary matrix.  As a result, the ZF transmit covariance matrices at BS $i$ can be expressed as
\begin{align}
\bar{\mv{S}}_{k} =  \bar p_{k}\frac{\hat{\mv{V}}_{k}\hat{\mv{V}}_{k}^H\mv{h}_{ik}^H\mv{h}_{ik}\hat{\mv{V}}_{k}\hat{\mv{V}}_{k}^H}{\|\mv{h}_{ik}\hat{\mv{V}}_{k}\|^2}, \label{eqn:svd:p3}
\end{align}
where $\bar p_{k} \ge 0, \forall k \in \mathcal{K}_i, i\in\mathcal{N}.$ Using (\ref{eqn:svd:p3}) and by defining $\bar a_{k} = \frac{\|\mv{h}_{ik}\hat{\mv{V}}_{k}\|^2}{\sigma_{k}^{2}} > 0$, $\forall k\in\mathcal{K}_i, i\in\mathcal{N}$, the weighted sum-rate maximization problem with energy cooperation only is formulated as the following joint power allocation and energy transfer problem.
\begin{align*}
\mathrm{(P2)}:\nonumber\\
\mathop{\mathtt{max}}\limits_{\{\bar p_k\},\{e_{ij}\}}&  \frac{1}{N}\sum\limits_{i\in\mathcal{N}}\sum\limits_{k\in\mathcal{K}_i}
\omega_k\log_2(1+\bar a_{k}\bar p_{k})\\
\mathtt{s.t.}~~
&\sum_{k\in\mathcal{K}_i}\bar p_{k}\le  E_i + \sum_{j\in\mathcal{N},j\neq i}\beta_{ji} e_{ji} - \sum_{j\in\mathcal{N},j\neq i}e_{ij},\forall i\in\mathcal{N}\\
&\bar{p}_{k}\ge 0, e_{ij}\ge 0,\ \forall k\in\mathcal{K}_i, ~i,j\in\mathcal{N}, i\neq j
\end{align*}

It is observed that (P2) is a special case of (P1), where $a_k$ and $p_k$ in (P1) is replaced as $\bar a_k$ and $\bar p_k$ in (P2), $\forall k\in\mathcal{K}$, respectively, and $\{b_{ik}\}$ in (P1) is set as $b_{ik}= 1,\forall k\in\mathcal{K}_i, i\in\mathcal{N}$, $b_{ik} = 0, \forall k \in \mathcal{K}\setminus\mathcal{K}_i, i\in\mathcal{N}$ in (P2). As a result, (P2) can be solved by Algorithm 1 by a change of variables/parameters as specified above. Note that although Algorithm 1 for solving (P2) may also need to be implemented at a central unit, it only requires the knowledge of $\bar a_{k}$'s, $\forall k\in\mathcal{K}_i$, which can be locally computed at each BS $i$. Thus, the complexity/overhead of solving (P2) for energy cooperation only is considerably lower than that of solving (P1) for joint communication and energy cooperation, which requires the additional transmit messages and CSI information from all BSs.

\subsection{No Cooperation}
When both communication and energy cooperation are not available, the scheme corresponds to problem (P2) by setting $\beta_{ij}=0,\forall i,j\in\mathcal{N},i\neq j$. In this case, we have $e_{ij}=0,\forall i,j\in\mathcal{N},i\neq j$, and Algorithm 1 can also be used to solve the weighted sum-rate maximization problem for this special case.

\section{Simulation Results}\label{sec:numerical}

In this section, we provide simulation results to evaluate the performance of our proposed joint communication and energy cooperation algorithm. We set $\omega_k = 1, \forall k\in\mathcal{K}$ and thus consider the sum-rate throughput of all MTs as the performance metric. We also set $\beta_{ij} \triangleq \beta, \forall i,j\in\mathcal{N}, i\neq j$.

First, we consider a simple two-cell network with single-antenna BSs to show the throughput gain of joint communication and  energy cooperation. We set $M=1, N=2, K=2$, and $\sigma_k^2 = 1, k=1,2$. It is assumed that the channels are independent and identically distributed (i.i.d.) Rayleigh fading, i.e., $h_{ik}$ is a CSCG random variable with zero mean and variance $\kappa_{ik}^2$, $i,k=1,2$. We further assume $\kappa_{11}^2=\kappa_{22}^2=1$ for the direct BS-MT channels, and $\kappa_{12}^2\le 1, \kappa_{21}^2\le1$ for the cross BS-MT channels. We average the sum-rates over 1000 random channel realizations.

\begin{figure}
\centering
 \epsfxsize=1\linewidth
    \includegraphics[width=8.5cm]{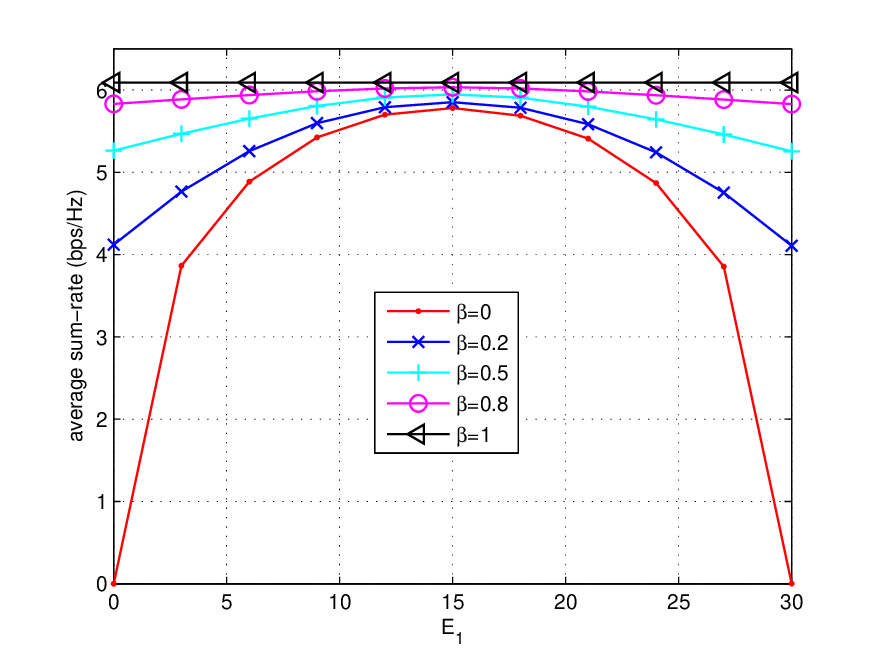}
\caption{Sum-rate performance with $E_1+E_2=30$ in a two-cell network.} \label{fig:4}
\end{figure}

In Fig. \ref{fig:4}, we consider the case when the energy arrival rates  $E_1$ and $E_2$ at the two BSs are constant over time subject to a given sum: $E_1 + E_2 = 30$.\footnote{The energy unit is normalized here such that one unit energy is equivalent to the transmit power required to have an average signal-to-noise ratio (SNR) at each MT equal to one or $0$ dB.} We set $\kappa_{12}^2 = \kappa_{21}^2=0.5$, and plot the average achievable sum-rate with different values of $\beta$ and $E_1$. It is observed that as $\beta$ increases, the sum-rate increases for any given $0\leq E_1\leq 30$, which is due to the fact that larger value of $\beta$ corresponds to smaller energy transfer loss. It is also observed that the sum-rate is zero when $\beta=0$ and $E_1 = 0$ or $30$ (accordingly, $E_2 = 30$ or $0$). This is because that in this case, there is one BS with zero transmit power, and thus without energy sharing between the two BSs  their cooperative ZF precoding is not feasible as discussed in Section \ref{sec:comp}. Furthermore, with any given $0 \le \beta<1$, it is observed that the maximum sum-rate is always achieved when the energy arrival rates at the two BSs are equal, i.e., $E_1 = E_2=15$. The reason is given as follows. Under the given symmetric channel setup, with equal $E_1$ and $E_2$, the amount of transferred energy  between the two BSs is minimized, and so is the energy loss in energy sharing; as a result, the total available energy for cooperative transmission is maximized, which thus yields the maximum sum-rate.

\begin{figure}
\centering
 \epsfxsize=1\linewidth
    \includegraphics[width=8.5cm]{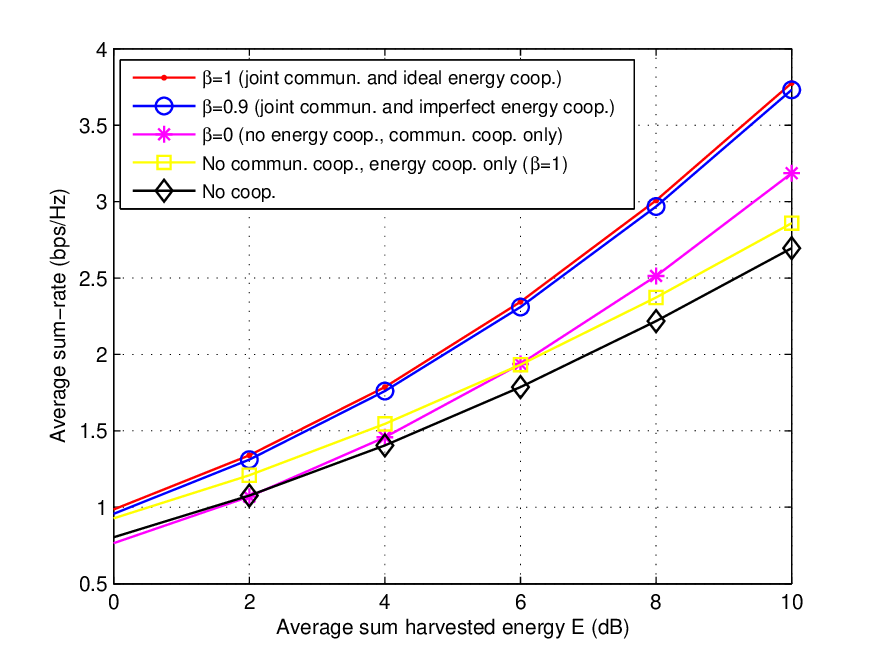}
\caption{Sum-rate comparison with versus without communication and/or energy cooperation in the low-SNR regime.} \label{fig:5}
\epsfxsize=1\linewidth
    \includegraphics[width=8.5cm]{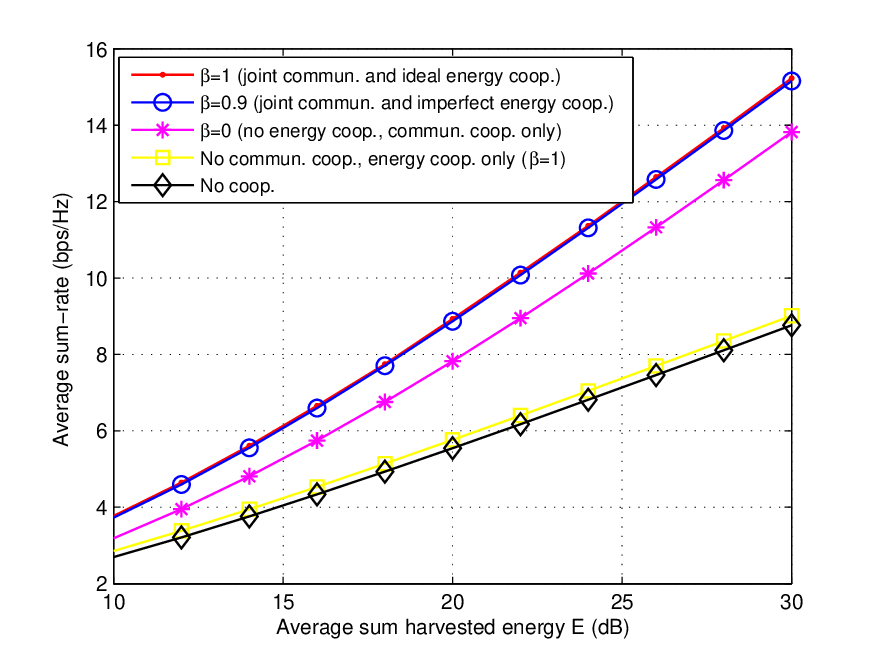}
\caption{Sum-rate comparison with versus without communication and/or energy cooperation in the high-SNR regime.} \label{fig:6}
\vspace{-0cm}
\end{figure}

In Figs. \ref{fig:5} and \ref{fig:6}, we show the sum-rate performance of the proposed optimal scheme with joint communication and energy cooperation with the practical energy transfer efficiency $\beta > 0$ as compared to three suboptimal schemes without communication  and/or energy cooperation given in Section \ref{sec:suboptimal}.{\footnote{
We set $\beta = 0.9$ as a practical energy transfer efficiency for the case with imperfect energy cooperation, since in practice, about 7-12\% of the electricity produced at the generation site is lost during the electricity transmission from the generation facilities to the end users (see \url{http://www.sunshineusainc.com/smartgrid.html}).}} We assume that $\kappa_{12}^2$ and $\kappa_{21}^2$ are independent and uniformly distributed in $[0,1]$, and the energy arrival rates $E_1$ and $E_2$ are independent and uniformly distributed in $[0,E]$ each with an equal mean of $\frac{E}{2}$, where $E$ denotes the average sum-energy harvested by both BSs. Note that the independent energy distribution may correspond to the case where the two BSs are powered by different renewable energy sources, e.g., one by solar energy and the other by wind energy. From Figs. \ref{fig:5} and \ref{fig:6}, it is observed that the joint communication and (ideal) energy cooperation with $\beta = 1$ always achieves the highest sum-rate, while the joint communication and (imperfect) energy cooperation with $\beta = 0.9$ achieves the sum-rate very close to that of $\beta = 1$, and also outperforms the other three suboptimal schemes without communication and/or energy cooperation. This shows the throughput  gain  of joint communication and energy cooperation, even with a non-negligible energy transfer loss. It is also observed that with $E<6$ dB, the scheme of ``energy cooperation only'' outperforms ``communication cooperation only''; however, the reverse is true when $E>6$ dB. This shows that the gain of energy cooperation is more dominant over that of communication cooperation in the low-SNR regime, but vice versa in the high-SNR regime.

\begin{figure}
\centering
 \epsfxsize=1\linewidth
    \includegraphics[width=8.5cm]{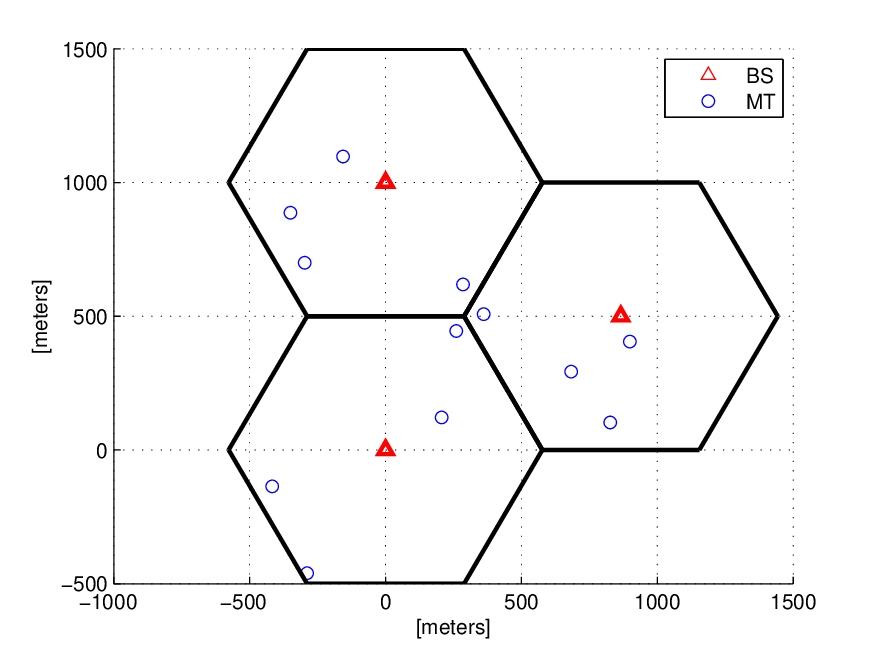}
\caption{Cell configuration for simulation.} \label{fig:7}
\end{figure}
\begin{figure}
\centering
 \epsfxsize=1\linewidth
    \includegraphics[width=8.5cm]{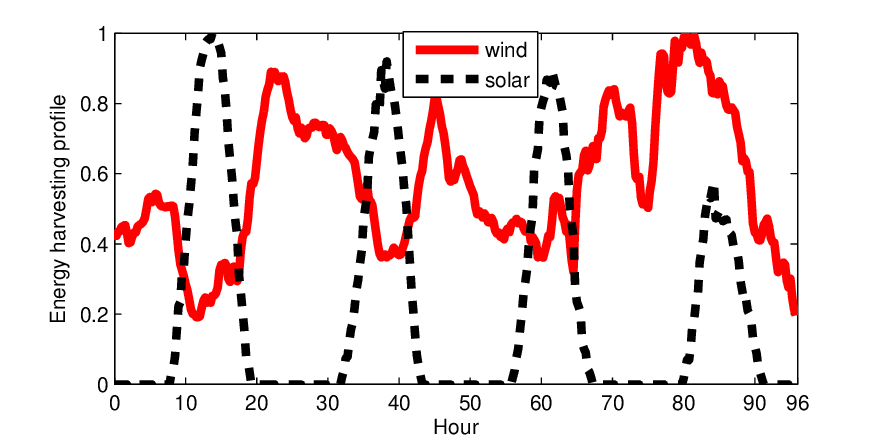}
\caption{Normalized solar and wind energy harvesting profiles.} \label{fig:8}
\end{figure}
Next, we evaluate the performances of the proposed joint communication and energy cooperation scheme as well as the other three suboptimal schemes in a practical three-cell cluster (with $N=3$) as shown in Fig. \ref{fig:7}, where the cells are hexagonal with the inter-BS distance of one kilometer, and we set $M=4$ and $K=12$. We further consider practical wind and solar energy profiles by using the real aggregated solar and wind generations from  0:00, 01 October 2013 to 0:00, 05 October 2013 in Belgium.{\footnote{See \url{http://www.elia.be/en/grid-data/power-generation/} for more details.}} Based on this real data, we can obtain the normalized wind and solar energy harvesting profiles over time as shown in Fig. \ref{fig:8}, where the energy harvesting rates are sampled (averaged) every 15 minutes, and thus the four days' renewable energy data corresponds to 384 points. Let the normalized wind and solar energy harvesting profiles in Fig. \ref{fig:8} be denoted as $\mv{\tau}_{\rm w} = [{\tau}_{{\rm w},1},\ldots,{\tau}_{{\rm w},384}]$ and $\mv{\tau}_{\rm s} = [{\tau}_{{\rm s},1},\ldots,{\tau}_{{\rm s},384}]$, respectively. In our simulations, we assume that all three BSs are deployed with both solar panels and wind turbines of different generation capacities, and set the energy harvesting rates over time at three BSs as $\mv{\tau}_1 = \bar E\cdot(0.5\mv{\tau}_{\rm w} + 0.5\mv{\tau}_{\rm s})$, $\mv{\tau}_2 = \bar E\cdot(0.1\mv{\tau}_{\rm w} + 0.9\mv{\tau}_{\rm s})$, and $\mv{\tau}_3 = \bar E\cdot(0.9\mv{\tau}_{\rm w} + 0.1\mv{\tau}_{\rm s})$, respectively, as shown in Fig. \ref{fig:9}, where $\bar E$ is a given constant. Moreover, since the channel coherence time is much shorter than the energy sample time of 15 minutes, for each sample of energy harvesting profiles we randomly generate $4$ MTs in the area covered by each hexagonal cell and run 100 independent channel realizations.

\begin{figure}
\centering
 \epsfxsize=1\linewidth
    \includegraphics[width=8.5cm]{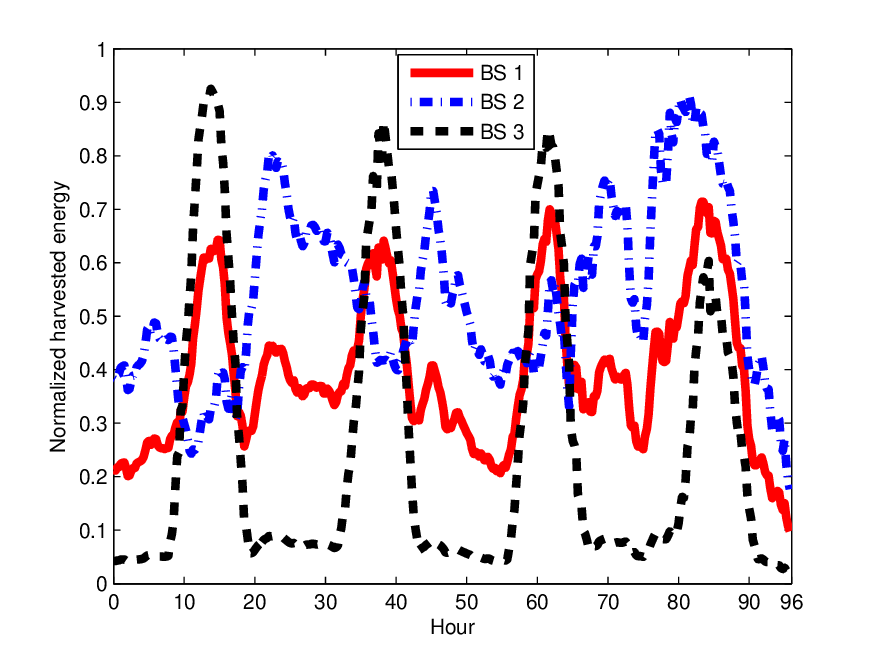}
\caption{Normalized energy harvesting profiles (by $\bar E$) at three BSs for simulation.} \label{fig:9}
\end{figure}

We assume that the channel is modelled by combining pathloss and short-term fading. For each channel realization, it is assumed that the distance-dependent channel attenuation from BS $i$ to MT $k$ is fixed, which is determined by the pathloss model of $\tilde{h}_{ik} = c_0\left(\frac{d_{ik}}{d_0}\right)^{-\varsigma}$, where $c_0 = -60$ dB is a constant equal to the pathloss at a reference distance $d_0 = 10$ meters, $\varsigma = 3.7$ is the pathloss exponent, and $d_{ik}$ denotes the distance from BS $i$ to MT $k$. We also consider i.i.d. Rayleigh distributed short-term fading, i.e., the channels for each realization, denoted by $\mv{h}_{ik}$'s, are zero mean CSCG random vectors with covariance matrix specified by the corresponding pathloss, i.e., $\tilde{h}_{ik}\mv{I}$. In addition, we assume that the background noise at each MT receiver  is $-85$ dBm. As a result, supposing the transmit power at each BS is $10$ dBW, the average SNR for each BS to a MT at each vertex of its covered hexagon is thus 0 dB.

\begin{figure}
\centering
 \epsfxsize=1\linewidth
    \includegraphics[width=8.5cm]{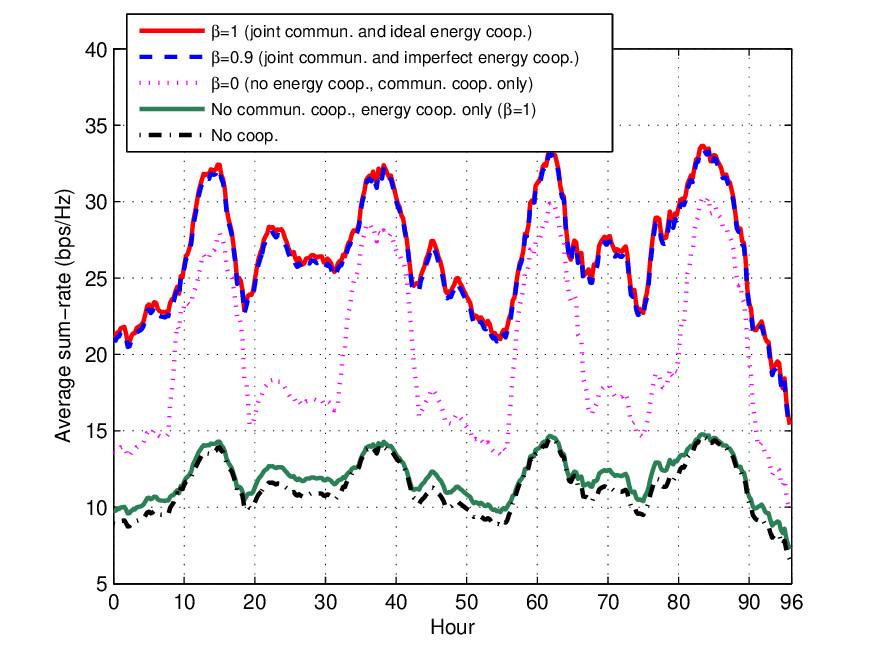}
\caption{Sum-rate performance over time, where $\bar E=10$ dBW.} \label{fig:10}
\end{figure}

In Fig. \ref{fig:10},  we show the average sum-rate performance over time, where we set $\bar E=10$ dBW. It is observed that the performance gain of the ``joint communication and energy cooperation'' scheme against the ``communication cooperation only'' scheme is very significant in the night time (e.g., hours 0-7, 20-30), while this gain becomes much smaller during the day time (e.g., hours 10-18). The reason is given as follows. In the night time, the sum-rate performance of the ``communication cooperation only'' scheme is limited by BS 3 that mainly relies on the solar generation and thus has very low energy harvesting rates; hence, energy cooperation is most beneficial during the night period. In contrast, during the day time when the energy harvesting rates at three BSs are more evenly distributed, the energy cooperation gain becomes less notable. Furthermore, the performance gain of the ``energy cooperation only'' scheme over the ``no cooperation'' scheme is observed to behave similarly. The result shows the benefit of energy cooperation again when the energy harvesting rates are unevenly distributed, even in the case without communication cooperation.

\begin{figure}
\centering
 \epsfxsize=1\linewidth
    \includegraphics[width=8.5cm]{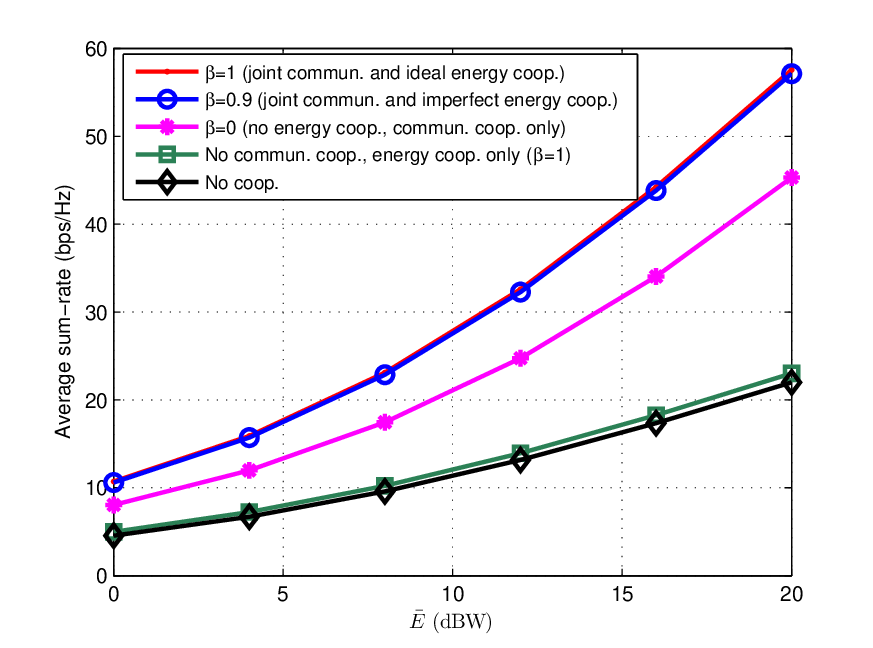}
\caption{Sum-rate comparison with versus without communication and/or energy cooperation.} \label{fig:11}
\end{figure}

In Fig. \ref{fig:11}, we show the MTs' average sum-rate performance over four days versus $\bar E$. Similar to the two-cell case in Figs. \ref{fig:5} and \ref{fig:6}, it is observed that the joint communication and energy cooperation with $\beta=1$ or $\beta=0.9$ considerably outperforms the other three suboptimal schemes at all values of $\bar E$, which shows the significant gains of joint communication and energy cooperation under this practical setup.

\section{Concluding Remarks}\label{sec:conclusion}

In this paper, we have proposed a new joint communication and energy cooperation approach for designing cellular networks with renewable power supplies. With the newly proposed {\color{black}aggregator-assisted} energy cooperation among cooperatively transmitting BSs, we formulate the weighted sum-rate maximization problem for a downlink CoMP system with ZF precoding. By applying convex optimization techniques, we develop an efficient solution to the optimal transmit power allocation and energy transfer at cooperative BSs. Furthermore, we show by simulations under practical setups the potential sum-rate gains by jointly exploiting communication and energy cooperation in cellular networks. It is revealed that under a practical energy transfer efficiency value, energy cooperation is most beneficial when the harvested energy among BSs is unevenly distributed. Due to the space limitation, there are several important issues unaddressed yet in this paper, which are briefly discussed in the following to motivate future work:
\begin{itemize}
 {\color{black} \item So far, we have assumed that the communication and energy information are perfectly obtained and shared within the cluster of BSs to characterize the theoretical limits of joint communication and energy cooperation. In practice, due to channel estimation error,  energy measurement error, and finite capacity and delay of backhaul links for information exchange etc., both communication and energy information at BSs can be imperfect. It is thus important to investigate their effect on the joint cooperation performance and design robust schemes for them.}
  \item We have considered one single cluster of cells with given users' association to BSs to simplify our study. However, in practice, cellular networks need to support time-varying traffic over different cells or clusters of cells. Therefore, how to optimally design communication and energy cooperation jointly with user association and/or cell clustering according to both renewable energy and traffic distributions in the network is also an interesting problem to investigate for future work.

 \item For the purpose of exposition, in this paper we have considered the energy cooperation in a single BS cluster operating in an energy self-sustainable manner (without purchasing any energy from the grid operator). In order to achieve more reliable quality of service (QoS) at MTs, it is beneficial to further allow the BS cluster to trade energy with the grid operator and/or other BS clusters. Given the different energy buying and selling prices provided by the grid operator, it is an interesting problem to jointly design the CoMP communication cooperation, the internal energy cooperation (within each BS cluster), and the external energy trading (with the grid operator and/or among BS clusters) for minimizing the total energy cost subject to the QoS requirements at MTs.
\item We have focused on the joint communication and energy cooperation over independent slots (in the time scale of communication scheduling, say, several milliseconds) to exploit the {\it spatial} renewable energy distributions. Alternatively, energy storage devices can be employed at BSs to charge/discharge energy over {\it time} to smooth out the random renewable energy to match the demand. Since the two techniques achieve the same goal of combating the random renewable energy over different dimensions (space versus time), they can be good complementarities and it is interesting to jointly design the storage management over time with the spatial communication and energy cooperation. However, since energy storage management is normally implemented in the same time scale of the energy harvesting process (say, several minutes), which is much larger than that of the joint communication and energy cooperation considered in this paper, the key challenge is how to efficiently implement the joint space and time cooperation by taking into account the time scale differences.
\end{itemize}

\appendices

\section{Proof of Lemma \ref{proposition:2}}\label{appendix:for:lemma1}
First, suppose that $\mu_{i}=0, \forall i\in\mathcal{N}$. In this case, it is easy to verify that the objective value of (\ref{eqn:15}) goes to infinity as $p_k \to\infty$ for any $k\in\mathcal{K}$, i.e., $f(\{{\mu}_i\})$ is unbounded from above. Therefore, $\mu_{i}, i\in\mathcal{N}$, cannot be all zero at the same time for $f(\{{\mu}_i\})$ to be bounded from above. The first part of this lemma is thus proved.

Next, suppose that there exists a pair of $i$ and $j$ satisfying $\beta_{ij}\mu_j-\mu_i > 0, i,j\in\mathcal{N},i\neq j$. In this case, it is easy to verify that the objective value of (\ref{eqn:15}) goes to infinity with $e_{ij} \to\infty$, i.e., $f(\{{\mu}_i\})$ is unbounded from above. Therefore, $\beta_{ij}\mu_j-\mu_i \le 0, \forall i,j\in\mathcal{N},i\neq j$, must be true for $f(\{{\mu}_i\})$ to be bounded from above. The second part of this lemma is thus proved. As a result, Lemma \ref{proposition:2} is proved.

\section{Proof of Proposition \ref{proposition:3.2}}\label{appendix:1}
Suppose that there exists at least an $\bar{i} \in \mathcal{N}$ such that $\mu_{\bar i}^* = 0$. Then according to $\beta_{ij}\mu_j^*-\mu_i^* \le 0, \forall i,j\in\mathcal{N}, i\neq j$ from Lemma \ref{proposition:2}, it follows that
\begin{align*}
\beta_{ij}\mu_j^*-\mu_{\bar i}^* = \beta\mu_j^* \le 0, \forall j\in\mathcal{N}, j\neq \bar{i}.
\end{align*}
Using the facts that $\beta_{ij}>0$ and $\mu_j^* \ge 0, \forall j\in\mathcal{N}$, it holds that $\mu_j^* = 0, \forall j\in\mathcal{N}, j\neq \bar{i}$. By combining this argument with the presumption that $\mu_{\bar i}^* = 0$, we have $\mu_i^* = 0, \forall i\in\mathcal{N}$, which contradicts Lemma \ref{proposition:2} that at least one of $\mu_{i}^*, i\in\mathcal{N}$, must be strictly positive. Therefore, $\mu_i^* > 0,\forall i\in\mathcal{N}$ must be true, which completes the proof of Proposition \ref{proposition:3.2}.

\section{Proof of Proposition \ref{proposition:3.3}}\label{appendix:2}
Suppose that there exists an optimal solution for (P1), denoted by $\{p_k^*,e_{ij}^*\}$ such that for one given BS $\bar i$ it holds that $\sum_{j\in\mathcal{N},j\neq \bar i}\beta_{j \bar i} e_{j \bar i}^*>0$ and $\sum_{j\in\mathcal{N},j\neq \bar i}e_{\bar ij}^*>0$, at the same time. Without loss of generality we can assume that there exist two BSs $\bar j$ and $\tilde j$ satisfying $e_{\bar j\bar i}^*>0$ and $e_{\bar i\tilde j}^*>0$ with $\bar j\neq \bar i$ and $\tilde j\neq \bar i$. In this case, by letting $A=\left(1- \frac{\beta_{\bar j\bar i}\beta_{\bar i\tilde j}}{\beta_{{\bar j\tilde j}}}\right)\min(e_{\bar j\bar i}^*,e_{\bar i\tilde j}^*)$ with $A > 0$ due to ${\beta_{\bar j\bar i}\beta_{\bar i\tilde j}} < {\beta_{{\bar j\tilde j}}}$, we can construct a new solution for (P2) as $\{\bar p_k,\bar e_{ij}\}$, where $\{\bar e_{ij}\}$ is given by
\begin{align}
\bar e_{\bar i\tilde j} = & e_{\bar i\tilde j}^*- \min(e_{\bar j\bar i}^*,e_{\bar i\tilde j}^*) + \frac{1}{N}A, \label{eqn:proof:2} \\
\bar e_{\bar j\tilde j} = & e_{\bar j\tilde j}^*+\frac{\beta_{\bar i\tilde j}}{\beta_{{\bar j\tilde j}}}\min(e_{\bar j\bar i}^*,e_{\bar i\tilde j}^*),\label{eqn:proof:3}\\
\bar e_{\bar j\bar i} =& e_{\bar j\bar i}^* - \frac{\beta_{\bar i\tilde j}}{\beta_{{\bar j\tilde j}}}\min(e_{\bar j\bar i}^*,e_{\bar i\tilde j}^*), \label{eqn:proof:1}\\
\bar e_{\bar i j} = & e_{\bar i j}^* + \frac{1}{N}A, \forall j\in\mathcal{N},j\neq \bar i, j\neq \tilde j,\label{eqn:proof:4}
\end{align}
and for other $\bar e_{ij}$'s not defined above, we have $\bar e_{ i j} = e_{ i j}^*$. Substituting $\{\bar e_{ij}\}$ defined in (\ref{eqn:proof:2})-(\ref{eqn:proof:4}) into (\ref{eqn:v4:14}), it follows after some simple manipulations that
\begin{align}
\sum_{k\in\mathcal{K}} b_{ik} {p}_{k}^*  =&  E_i - \sum_{j\in\mathcal{N},j\neq i}\beta \bar e_{ji} + \sum_{j\in\mathcal{N},j\neq i} \bar e_{ij} - \frac{1}{N}A\nonumber\\
 <&E_i - \sum_{j\in\mathcal{N},j\neq i}\beta \bar e_{ji} + \sum_{j\in\mathcal{N},j\neq i} \bar e_{ij},~~~~~~ \forall i\in\mathcal{N}.
\end{align}
As a result, we can find a set of $\bar p_k$'s with $\bar p_k = p_k^* + \delta> p_k^*\ge 0, \forall k\in\mathcal{K},$ where $\delta>0$ is chosen as sufficiently small such that
\begin{align}
 \sum_{k\in\mathcal{K}} b_{ik} \bar{p}_{k}  \le E_i - \sum_{j\in\mathcal{N},j\neq i}\beta \bar e_{j i} + \sum_{j\in\mathcal{N},j\neq i} \bar e_{ij}, \forall i\in\mathcal{N}.
\end{align}
Therefore, the newly constructed $\{\bar p_k,\bar e_{ij}\}$ is a feasible solution of (P1) and achieves a larger objective value than that by $\{p_k^*,e_{ij}^*\}$, which contradicts the presumption that $\{p_k^*,e_{ij}^*\}$ is optimal. Thus, it holds that at least one of $\sum_{j\in\mathcal{N},j\neq i}\beta e_{ji}^*$ and $\sum_{j\in\mathcal{N},j\neq i}e_{ij}^*$ should be zero, $\forall i\in\mathcal{N}$. Proposition \ref{proposition:3.3} is thus proved.


\begin{thebibliography}{1}
\bibliographystyle{IEEEbib}
\bibitem{Hasan}
Z. Hasan, H. Boostanimehr, and V. Bhargava, ``Green cellular networks: A survey, some research issues and challenges,'' {\it IEEE Commun.
Surveys \& Tutorials}, vol. 13, no. 4, pp. 524-540, 2011.

\bibitem{huawei}
Huawei, ``Mobile networks go green,'' available online at \url{http://www.huawei.com/en/about-huawei/publications/communicate/hw-082734.htm}

\bibitem{HanAnsari2014}
T. Han and N. Ansari, ``Powering mobile networks with green energy,'' {\it IEEE Wireless Commun.}, vol. 21, no. 1, pp. 90-96, Feb. 2014.

{\color{black}
\bibitem{XueSmartGird}
X. Fang, S. Misra, G. Xue, and D. Yang, ``Smart grid - the new and improved power grid: a survey,'' {\it IEEE Commun. Surveys \& Tutorials}, vol. 14, no. 4, pp. 944-980, 2012.}

\bibitem{Leithon2013}
J. Leithon, T. J. Lim, and S. Sun, ``Online energy management strategies for base stations powered by the smart grid,'' in {\it Proc. IEEE SmartGridComm}, pp. 199-204, Oct. 2013.


\bibitem{Chen2013}
S. Chen, N. B. Shroff, and P. Sinha, ``Energy trading in the smart grid: from end-user's perspective,'' in {\it Proc. Asilomar Conference on Signals, Systems \& Computers}, Nov. 2013.

\bibitem{Ilic}
D. Ilic, P. G. D. Silva, S. Karnouskos, and M. Griesemer, ``An energy market for trading electricity in smart grid neighbourhoods,'' in {\it Proc. 6th IEEE International Conference on Digital Ecosystems Technologies (DEST)}, pp. 1-6, Jun. 2012.

\bibitem{SaadHan}
W. Saad, Z. Han, H. V. Poor, and T. Ba\c sar, ``Game-theoretic methods for the smart grid: an overview of microgrid systems, demand-side management, and smart grid communications,'' {\it IEEE Sig. Process. Mag.}, vol. 29, no. 5, pp. 86-105, Sep. 2012.

\bibitem{Gesbert2010}
D. Gesbert, S. Hanly, H. Huang, S. Shamai, O. Simeone, and W. Yu, ``Multi-cell MIMO cooperative networks: A new look at interference,''
{\it IEEE J. Sel. Areas Commun.}, vol. 28, no. 9, pp. 1380-1408, Dec. 2010.

\bibitem{Irmer2011}
R. Irmer, H. Droste, P. Marsch, M. Grieger, G. Fettweis, S. Brueck, H. Mayer, L. Thiele, and V. Jungnickel, ``Coordinated multipoint: Concepts, performance,
and field trial results,'' {\it IEEE Commun. Mag.}, vol. 49, no. 2,
pp. 102-111, Feb. 2011.


\bibitem{NgHuang2010}
C. T. K. Ng and H. Huang, ``Linear precoding in cooperative MIMO cellular networks with limited coordination clusters,''
{\it IEEE  J. Sel. Areas Commun.}, vol. 28, no. 9, pp. 1446-1454, Dec. 2010.

\bibitem{ZhangChen2009}
J. Zhang, R. Chen, J. Andrews, A. Ghosh, and R. W. Heath, ``Networked
MIMO with clustered linear precoding,'' {\it IEEE Trans. Wireless Commun.},
vol. 8, no. 4, pp. 1910-1921, Apr. 2009.

{\color{black}\bibitem{Gkatzikis2013}
L. Gkatzikis, I. Koutsopoulos, and T. Salonidis, ``The role of aggregators in smart grid demand response markets,'' {\it IEEE J. Sel. Areas Commun.}, vol. 31, no. 7, pp. 1247-1257, Jul. 2013.}

\bibitem{Zhang2010}
R. Zhang, ``Cooperative multi-cell block diagonalization with per-base-station power constraints,'' {\it IEEE J. Sel. Areas Commun.},  vol. 28, no. 9, pp. 1435-1445, Sep. 2010.

\bibitem{Chia2013}
Y. K. Chia, S. Sun, and R. Zhang, ``Energy cooperation in cellular networks with renewable powered base stations,'' in {\it Proc. IEEE WCNC}, pp. 2542-2547, Apr. 2013.

\bibitem{GuoTCOM}
Y. Guo, J. Xu, L. Duan, and R. Zhang, ``Optimal energy and spectrum sharing for cooperative cellular systems,'' in {\it Proc. 2014 IEEE ICC}, Jun. 2014.

\bibitem{Zheng2013}
M. Zheng, P. Pawelczak, S. Sta\'nczak, and H. Yu, ``Planning of cellular networks enhanced by energy harvesting,'' {\it IEEE Commun. Letters}, vol. 17, no. 6, pp. 1092-1095, Jun. 2013.

\bibitem{GongNiu2013}
C. Hu, X. Zhang, S. Zhou, and Z. Niu, ``Utility optimal scheduling in energy cooperation networks powered by renewable energy,'' in {\it Proc. IEEE Asia-Pacific Conference on Communications (APCC),} pp. 403-408, Aug. 2013.


\bibitem{Gurakan2012}
B. Gurakan, O. Ozel, J. Yang, and S. Ulukus, ``Energy cooperation in energy harvesting communications,'' {\it IEEE Trans. Commun.}, vol. 61, no. 12, pp. 4884-4898, Dec. 2013.

\bibitem{Zheng2014}
G. Zheng, Z. Ho, E. A. Jorswieck, and  B. Ottersten, ``Information and energy cooperation in cognitive radio networks,'' {\it IEEE Trans. Sig. Process.}, vol. 62, no. 9, pp. 2290-2303, May 2014.

\bibitem{BuYu2012}
S. Bu, F. R. Yu, Y. Cai, and X. P. Liu, ``When the smart grid meets energy-efficient communications: Green wireless cellular networks powered by the smart grid,'' {\it IEEE Trans. Wireless Commun.}, vol. 11, no. 8, pp. 3014-3024, Aug. 2012.


\bibitem{YooGoldsmith}
T. Yoo and A. Goldsmith, ``On the optimality of multi-antenna broadcast scheduling using zero-forcing beamforming,'' {\it IEEE J. Sel. Areas Commun.,} vol. 24, no. 3, pp. 528-541, Mar. 2006.

\bibitem{WoodWollenberg}
A. J. Wood and B. F. Wollenberg, {\it Power Generation, Operation, and Control}, Wiley, 1996.


\bibitem{BoydVandenberghe2004} S. Boyd and L. Vandenberghe, {\it Convex Optimization}, Cambridge
University Press, 2004.





\bibitem{Boyd:ConvexII}
S. Boyd, ``Convex optimization II,'' Stanford University. (available online at {\url{http://www.stanford.edu/class/ee364b/lectures.html}})

\bibitem{cvx}
M. Grant and S. Boyd, {\it CVX: Matlab software for disciplined convex programming}, version 1.21, http://cvxr.com/cvx/, Apr. 2011.

\end{thebibliography}
\end{document}